\documentclass[11pt]{article}

\usepackage[margin = .97in]{geometry}

\usepackage{enumerate}
\usepackage{amsmath,amssymb,amsthm,hyperref,color}
\usepackage{amsbsy}
\hypersetup{colorlinks=true,citecolor=blue, linkcolor=blue, urlcolor=blue}

\newtheorem{theorem}{Theorem}
\newtheorem{lemma}[theorem]{Lemma}
\newtheorem{claim}[theorem]{Claim}
\newtheorem{proposition}[theorem]{Proposition}
\newtheorem{observation}[theorem]{Observation}
\newtheorem{corollary}[theorem]{Corollary}

\newtheorem{condition}{Condition}

\def\I{\mathcal{I}}
\def\G{\mathcal{G}}

\def\E{\mathbb{E}}
\def\Tree{\mathbb{T}}

\def\TD{\mathcal{R}_{\Delta}}
\def\eps{\varepsilon}
\def\epsilon{\varepsilon}
\newcommand{\Tmix}{T_{\mathrm{mix}}}

\newcommand{\rpe}{p^+}
\newcommand{\rpo}{p^-}

\title{Improved Inapproximability Results for Counting
Independent Sets in the Hard-Core Model}
\author{Andreas Galanis\thanks{School of Computer Science, Georgia
Institute of Technology, Atlanta GA 30332.
Email: \{agalanis,vigoda,ljyang\}@cc.gatech.edu.
Research supported in part by NSF grant CCF-0830298 and CCF-0910584.}
\and
 Qi Ge\thanks{
Department of Computer Science, University of Rochester,
Rochester, NY 14627.  Email: \{qge,stefanko\}@cs.rochester.edu.
Research supported in part by NSF grant CCF-0910415.}
 \and
 Daniel \v{S}tefankovi\v{c}$^\dag$
 \and Eric Vigoda$^\star$
 \and Linji Yang$^\star$}

\begin{document}

\maketitle

\begin{abstract}
We study the computational complexity of approximately counting the number of independent sets of a graph with maximum degree
$\Delta$.  More generally, for an input graph $G=(V,E)$ and an activity $\lambda>0$, we are interested in the quantity
$Z_G(\lambda)$ defined as the sum over independent sets $I$ weighted as $w(I) = \lambda^{|I|}$.  In statistical
physics, $Z_G(\lambda)$ is the partition function for the hard-core model, which is an idealized model of a gas where the particles have non-negligible size.
  Recently, an interesting phase transition was shown
to occur for the complexity of approximating the partition function.
Weitz showed an FPAS for the partition function for any graph
of maximum degree $\Delta$ when $\Delta$ is
constant and $\lambda<\lambda_c(\Tree_\Delta):=(\Delta-1)^{\Delta-1}/(\Delta-2)^\Delta$.
The quantity $\lambda_c(\Tree_\Delta)$ is the critical point for
the so-called uniqueness threshold on the infinite, regular tree of
degree $\Delta$.  On the other side,
Sly proved that there does not exist efficient (randomized) approximation
algorithms for
$\lambda_c(\Tree_\Delta)<\lambda<\lambda_c(\Tree_\Delta)+\varepsilon(\Delta)$,
unless NP$=$RP, for some function $\varepsilon(\Delta)>0$.
 We remove the upper bound in the assumptions of Sly's result for $\Delta\neq 4,5$,
that is, we show that there does not exist efficient randomized approximation algorithms for all $\lambda>\lambda_c(\Tree_\Delta)$
for $\Delta=3$ and $\Delta\geq 6$.
Sly's inapproximability result uses a clever reduction,
combined with a second-moment analysis
of Mossel, Weitz and Wormald which prove torpid mixing of the
Glauber dynamics for sampling from the associated Gibbs distribution
on almost every regular graph of degree $\Delta$ for
the same range of $\lambda$ as in Sly's result.  We extend Sly's result
by improving upon the technical work of Mossel et al.,
via a
more detailed analysis of independent sets in random regular graphs.
\end{abstract}

%\thispagestyle{empty}
%\newpage
%\setcounter{page}{1}

\section{Introduction}
For a graph $G=(V,E)$ and activity $\lambda>0$, the
hard-core model is defined on the set $\I(G)$ of independent sets
of $G$ where set $I\in\I(G)$ has weight $w(I):=\lambda^{|I|}$.
The so-called partition function for the model is defined as:
\[  Z_G(\lambda)  := \sum_{I\in \I(G)} w(I)
  = \sum_{I\in \I(G)} \lambda^{|I|}.
\]
The Gibbs distribution $\mu$ is over the set $\I(G)$ where
$\mu(I)=w(I)/Z_G(\lambda)$.
The case $\lambda=1$ is especially interesting from a combinatorial
perspective, since the partition function
is the number of independent sets
in $G$ and the Gibbs distribution is uniformly distributed over the
set of independent sets.

The hard-core model has received considerable attention
in several fields.
In statistical physics, it is studied as an idealized model of
a gas where the gas particles have non-negligible size so
neighboring sites cannot simultaneously be occupied
\cite{GF,BS}.  The activity
$\lambda$ corresponds to the fugacity of the gas.
The model also arose in operations research in the study of
communication networks \cite{Kelly}.

We study the computational complexity of approximating the partition
function.
Valiant \cite{Valiant} proved that exactly computing the number of independent
sets of an input graph $G=(V,E)$ is \#P-complete.
Greenhill \cite{Greenhill} proved that even when the input is
restricted to graphs with
maximum degree $3$, it is still \#P-complete.
Hence, our focus is on approximating the partition function.

Weitz~\cite{Weitz} gave an FPAS (fully polynomial-time approximation
scheme) for the partition function of graphs with maximum
degree $\Delta$ when $\Delta$ is constant and
$\lambda<\lambda_c(\Tree_\Delta):= (\Delta-1)^{\Delta-1}/(\Delta-2)^\Delta$.
The activity $\lambda_c(\Tree_\Delta)$
is the critical activity for the threshold of uniqueness/non-uniqueness
of the infinite-volume Gibbs measures on the infinite
$\Delta$-regular tree \cite{Kelly}.
Recently, Sly \cite{Sly10} proved that, unless $NP=RP$,
 for every $\Delta \geq 3$, there exists a function $\epsilon(\Delta)>0$
 such that for graphs with maximum degree $\Delta$
 there does not exist an
 FPRAS (fully-polynomial time randomized approximation scheme)
 for the partition function at activity $\lambda$ satisfying:
 \begin{equation}
 \label{eps-constraint}
\tag{$\star$}
\lambda_c(\Tree_\Delta)<\lambda<\lambda_c(\Tree_\Delta)+\eps(\Delta).
\end{equation}
It was conjectured in Sly \cite{Sly10} and Mossel et al. \cite{MWW}
that the inapproximability result holds
for all $\lambda> \lambda_c(\Tree_\Delta)$.
We almost resolve this conjecture, that is we prove the
conjecture for all $\Delta$ with the exception of
$\Delta\in\{4,5\}$.

\begin{theorem}\label{thm:main}
Unless NP$=$RP, there does not exist an FPRAS for the partition function of the hard-core model
for graphs of maximum degree at most $\Delta$
at activity $\lambda$ when:
\begin{itemize}
\item $\Delta = 3$ and $\lambda > \lambda_c(\Tree_3)=4$; or
\item $\Delta\geq 6$ and $\lambda > \lambda_c(\Tree_\Delta)$; or
\item $\Delta = 4$ and $\lambda \in (\lambda_c(\Tree_4)=1.6875, 2.01538] \cup (4,+\infty)$; or
\item $\Delta = 5$ and $\lambda \in (\lambda_c(\Tree_5)=256/243, 1.45641] \cup (1.6875,2.01538] \cup (4,+\infty)$.
\end{itemize}
\end{theorem}

Sly's work utilizes earlier work of Mossel et al. \cite{MWW}
which studied the Glauber dynamics.
The Glauber dynamics is a simple Markov chain $(X_t)$ that is
used to sample from the Gibbs distribution
(and hence to approximate the partition function via now
standard techniques, see \cite{JVV,SVV}).
For an input graph $G=(V,E)$ and activity $\lambda>0$,
the state space of the chain is $\I(G)$.
From a state $X_t\in\I(G)$, the transitions $X_t\rightarrow X_{t+1}$
are defined by the following stochastic process:
\begin{itemize}
\item
Choose a vertex $v$ uniformly at random from $V$.
\item
Let \[
X' = \begin{cases}
X_t\cup\{v\} & \mbox{with probability } \lambda/(1+\lambda)
\\
X_t\setminus\{v\} & \mbox{with probability } 1/(1+\lambda).
\end{cases}
\]
\item
If $X'\in\I(G)$, then set $X_{t+1} = X'$, otherwise
set $X_{t+1}=X_t$.
\end{itemize}
It is straightforward to verify that the Glauber dynamics is ergodic,
and the unique stationary distribution is the Gibbs distribution.
The mixing time $\Tmix$ is the minimum number of steps $T$
from the worst initial state $X_0$,
so that the distribution of $X_T$ is within (total) variation distance
$\leq 1/4$ of the stationary distribution.
The chain is said to be {\em rapidly mixing} if the mixing time
is polynomial in $n=|V|$, and it is said to be {\em torpidly mixing}
if the mixing time is exponential in $n$ (for the purposes of this
paper, that means $\Tmix =  \exp(\Omega(n))).$
We refer the reader to Levin et al. \cite{LPW} for a more thorough
introduction to the Glauber dynamics.

Mossel et al. \cite{MWW}
proved that the Glauber dynamics is torpidly mixing,
for all $\Delta\geq 3$, for graphs with maximum degree $\Delta$
when $\lambda$ satisfies \eqref{eps-constraint}.
This result improved upon earlier work of Dyer et al. \cite{DFJ} which
held for larger $\lambda$, but not down to the critical activity
$\lambda_c(\Tree_\Delta)$.
The torpid mixing result of Mossel et al. \cite{MWW}
follows immediately (via a conductance argument)
from their main result that for a random $\Delta$-regular bipartite graph, for
$\lambda$ satisfying~\eqref{eps-constraint}, an independent
set drawn from the Gibbs distribution is ``unbalanced'' with high
probability.

The proof of Mossel et al. \cite{MWW}
is a technically involved second moment calculation
that Sly \cite{Sly10} calls a ``technical tour de force''.
Our main contribution is to improve upon Mossel et al.'s
result, most notably, extending it to all $\lambda>\lambda_c(\Tree_\Delta)$ for
$\Delta=3$. Our improved analysis comes from
using a slightly different parameterization of the second moment
of the partition function, which brings in symmetry, and
allows for simpler proofs.

To formally state our results for independent sets of random
regular graphs, we need to partition the
 set of independent sets as follows.
 For a bipartite graph $G=(V_1\cup V_2,E)$ where $|V_1|=|V_2|=n$,
  for $\delta>0$,  for $i\in\{1,2\}$, let
 \[  \I_i^\delta = \{ I \in\I(G): |I\cap V_i| > |I\cap V_{3-i}| + \delta n\}
 \]
 denote the independent sets that are unbalanced and
 ``biased'' towards $V_i$.  Let
 \[ \I_B^\delta =  \{I\in\I(G): |I\cap V_i| \leq |I\cap V_{3-i}| + \delta n\} \]
 denote the set of nearly balanced independent sets.

 Let $\G(n,\Delta)$
denote the probability distribution over bipartite graphs with $n+n$
vertices formed by taking the union of $\Delta$ random perfect matchings.
Strictly speaking, this distribution is over multi-graphs.  However,
for independent sets the multiplicity of an edge does not matter
so we can view $\G(n,\Delta)$ as a distribution over simple graphs
with maximum degree $\Delta$.
Moreover, since our results hold asymptotically almost surely (a.a.s.) over
$\G(n,\Delta)$, as noted in \cite[Section 2.1]{MWW},
by standard arguments (see the note after the proof
of \cite[Theorem 4]{MRRW}), our
results also hold a.a.s. for the uniform distribution over bipartite
$\Delta$-regular graphs.

\begin{theorem}
\label{thm:unbalanced-3}
For $\Delta=3$ and any $\lambda > \lambda_c(\Tree_\Delta)$, there exist
$a>1$ and $\delta>0$
such that, asymptotically almost surely,
for a graph $G$ chosen from $\G(n,\Delta)$, the Gibbs distribution $\mu$ satisfies:
\begin{equation*}
\mu(\I_B^\delta) \leq a^{-n}\min\{\mu(\I_1^\delta), \mu(\I_2^\delta)\}.
\end{equation*}
Therefore, the Glauber dynamics is torpidly mixing.
\end{theorem}

This proves Conjecture 2.4 of \cite{MWW} for the case $d=3$.
For $\Delta\geq 4$ we can extend Mossel et al.'s results to a
larger range of $\lambda$ than previously known.

\begin{theorem}
\label{thm:unbalanced-bigger}
For all $\Delta\geq 4$ there exists $\eps(\Delta)>0$,
for any $\lambda$ where $\lambda_c(\Tree_\Delta)<\lambda<\lambda_c(\Tree_\Delta)+\eps(\Delta)$,
there exist $a>1$ and $\delta>0$
such that, asymptotically almost surely,
for a graph $G$ chosen from $\G(n,\Delta)$, the Gibbs distribution $\mu$ satisfies:
\begin{equation*}
  \mu(\I_B^\delta) \leq a^{-n}\min\{\mu(\I_1^\delta), \mu(\I_2^\delta)\}.
\end{equation*}
Therefore, the Glauber dynamics is torpidly mixing.
The function $\eps(\Delta)$ satisfies:
\begin{itemize}
\item If $\Delta\geq 6$,
  $\eps(\Delta)\geq  \lambda_c(\Tree_\Delta)-\lambda_c(\Tree_{\Delta+1})$.
  \item For $\Delta=5$, $\eps(5) \geq .402912$.
    \item For $\Delta=4$, $\eps(4) \geq .327887$.
    \end{itemize}
\end{theorem}

In the following section we look at the first and second moments
of the partition function.  We then state two technical lemmas (Lemma
\ref{lem:main} and Lemma
\ref{lem:d3}) from which Theorems \ref{thm:main},
\ref{thm:unbalanced-3} and \ref{thm:unbalanced-bigger}
easily follow using work of Sly \cite{Sly10} and Mossel et al. \cite{MWW}.
In Section \ref{sec:analysis} we prove the technical lemmas.
Some of our proofs use Mathematica to prove inequalities involving
rational functions, this is discussed in Section \ref{sec:computer-assist}.

\section{Overview}

Proceeding as in \cite{MWW} and \cite{Sly10}, roughly speaking,
to prove Theorems \ref{thm:main}, \ref{thm:unbalanced-3} and \ref{thm:unbalanced-bigger}
we need to prove that there exist graphs $G$ whose
partition function is close to the expected value (where
the expectation is over a random $G\sim\G(n,\Delta)$). At the heart of the argument
lies a careful analysis of the first two moments of the partition function.
The aim of this Section is to give a brief technical overview of the analysis.
For more details, the reader is referred to \cite{MWW} and \cite{Sly10}.

In Section~\ref{sec:phtr}, we define the quantities which will be prevalent
throughout the text. In Sections~\ref{sec:firstmom} and~\ref{sec:secondmom},
we revisit the first and second moments and state our
main technical Lemma.  We then prove Theorems \ref{thm:main},
\ref{thm:unbalanced-3} and \ref{thm:unbalanced-bigger}
in  Section~\ref{sec:contri}.
In Section~\ref{sec:computer-assist},
we clarify our use of computer assistance for the proofs of some technical inequalities.

\subsection{Phase Transition Revisited}\label{sec:phtr}

Recall, for the infinite $\Delta$-regular tree $\Tree_\Delta$,
Kelly \cite{Kelly} showed that
there is a phase transition at $\lambda_c(\Tree_\Delta) = (\Delta-1)^{\Delta-1}/(\Delta-2)^\Delta$.
Formally, this phase transition can be defined in the following
manner.  Let $T_\ell$ denote the complete tree of degree $\Delta$ and
containing $\ell$ levels.  Let $p_\ell$ denote the marginal probability
that the root is occupied in the Gibbs distribution on $T_\ell$.
Note, for even $\ell$ the root is in the maximum independent set,
whereas for odd $\ell$ the root is not in the maximum independent set.
Our interest is in comparing the marginal probability for the root in
even versus odd sized trees.   Hence, let
\[
 p^+ = \lim_{\ell\rightarrow\infty} p_{2\ell}
\ \ \mbox{ and } \ \
 p^- = \lim_{\ell\rightarrow\infty} p_{2\ell+1}.
\]
One can prove these limits exist by analyzing appropriate recurrences.
The phase transition on the tree $\Tree_\Delta$ captures whether
$p^+$ equals (or not) $p^-$.
For all $\lambda\leq \lambda_c(\Tree_\Delta)$, we have $p^+ = p^-$, and
let $p^*:=p^+=p^-$.
  On the other side, for all $\lambda>\lambda_c(\Tree_\Delta)$,
  we have $\rpe > \rpo$.  Mossel et al. \cite{MWW} exhibited
  the critical role these quantities $p^+$ and $p^-$ play in
  the analysis of the Gibbs distribution on random graphs $\G(n,\Delta)$.

\subsection{First Moment of the Partition Function}\label{sec:firstmom}
Let $G\sim \G(n,\Delta)$. For $\alpha,\beta\geq 0$, let
\[ \I^{\alpha,\beta}_G = \{I \in \I_G \,|\, |I \cap V_1| = \alpha n, |I \cap V_2| = \beta n \},
\]
that is, $\alpha$ is the fraction of the vertices in $V_1$ that are in the independent set
and $\beta$ is the fraction of the vertices in $V_2$ for an independent set of a bipartite graph $G$. We consider only values of $(\alpha,\beta)$ in the region
\begin{equation*}
\mathcal{R} = \{(\alpha,\beta) \,|\, \alpha,\beta \geq 0 \mbox{ and } \alpha+\beta\leq 1\},
\end{equation*}
since it is straightforward to see that for a graph $G\sim \G(n,\Delta)$, it \textit{deterministically} holds that $\I^{\alpha,\beta}_G=\emptyset$ whenever $(\alpha,\beta)\notin \mathcal{R}$. For $(\alpha,\beta)\in\mathcal{R}$, define also
\[ Z^{\alpha,\beta}_G = \sum_{I \in \I^{\alpha,\beta}_G} \lambda^{(\alpha+\beta)n},
\]
i.e., $Z^{\alpha,\beta}_G$ is the total weight of independent sets in $\I^{\alpha,\beta}_G$. The first moment of $Z^{\alpha,\beta}_G$ is
\begin{eqnarray*}
\E_{\G}[Z^{\alpha,\beta}_G] = \lambda^{(\alpha+\beta)n}\binom{n}{\alpha n}\binom{n}{\beta n}\left(\frac{\binom{(1-\beta)n}{\alpha n}}{\binom{n}{\alpha n}}\right)^{\Delta}
\approx \exp(n \Phi_1(\alpha,\beta)),
\end{eqnarray*}
where
$$
\Phi_1(\alpha,\beta) = (\alpha+\beta)\ln(\lambda)+H(\alpha)+H(\beta)+ \Delta(1-\beta)H\left(\frac{\alpha}{1-\beta}\right)-\Delta H(\alpha),
$$
and $H(x)= - x\ln x-(1-x)\ln(1-x)$ is the entropy function. The asymptotic order of $\E_{\G}[Z^{\alpha,\beta}_G]$ follows easily by Stirling's approximation.

The first moment was analyzed in the work of Dyer et al. \cite{DFJ}.
We use the following lemma from Mossel et al. \cite{MWW} that
relates the properties of the first moment to $p^*,p^+$ and $p^-$.
The most important aspect of this lemma is that in the non-uniqueness region
$\E_{\G}[Z^{\alpha,\beta}_G]$ is maximized when $\alpha\neq \beta$ (\cite{DFJ}) and more specifically for $(\alpha,\beta)=(p^+,p^-)$ (\cite{MWW}).

\begin{lemma}[Lemma 3.2 and Proposition 4.1 of Mossel et al. \cite{MWW}]\label{lemam}
The following holds:
\begin{enumerate}
\item the stationary point $(\alpha,\beta)$ of $\Phi_1$ over $\mathcal{R}$ is the solution to $\beta=\phi(\alpha)$ and $\alpha=\phi(\beta)$, where
\begin{equation}\label{eq:occupy}
\phi(x) = (1-x)\left(1-\left(\frac{x}{\lambda(1-x)}\right)^{1/\Delta}\right),
\end{equation}
and the solutions are exactly $(p^+,p^-)$, $(p^-,p^+)$, and $(p^*,p^*)$ when $\lambda>\lambda_c(\Tree_\Delta)$, and the unique solution is $(p^*,p^*)$ when $\lambda \leq \lambda_c(\Tree_\Delta)$;
\item when $\lambda \leq \lambda_c(\Tree_\Delta)$, $(p^*,p^*)$ is the unique maximum of $\Phi_1$ over $\mathcal{R}$, and when $\lambda > \lambda_c(\Tree_\Delta)$, $(p^+,p^-)$ and $(p^-,p^+)$ are the maxima of $\Phi_1$ over $\mathcal{R}$, and $(p^*,p^*)$ is not a local maximum;
\item all local maxima of $\Phi_1$ satisfy $\alpha+\beta+\Delta(\Delta-2)\alpha\beta \leq 1$;
\item $p^+,p^-,p^*$ satisfy $p^-<p^*<p^+$ and when $\lambda \to \lambda_c(\Tree_\Delta)$ from above, we have $p^*,p^-,p^+ \to 1/\Delta$.
\end{enumerate}
\end{lemma}

For every $\Delta \geq 3$, define the region
\begin{equation*}
\TD = \{(\alpha,\beta) \,|\, \alpha,\beta > 0 \mbox{ and } \alpha+\beta+\Delta(\Delta-2)\alpha\beta \leq 1\}.
\end{equation*}
Part 3 of Lemma~\ref{lemam} establishes that the local maxima of $\Phi_1$ lie in $\TD$.
Note that for $\Delta \geq 3$, we have $\TD \subset \mathcal{R}$.
Hence, the local maxima for all $\Delta\geq 3$ lie in the interior of $\mathcal{R}$.

\subsection{Second Moment of the Partition Function}\label{sec:secondmom}

The second moment of $Z^{\alpha,\beta}_G$ satisfies (\cite{MWW})
\begin{eqnarray*}
\E_{\G}[(Z^{\alpha,\beta}_G)^2] &\approx& \exp(n \cdot \max_{\gamma,\delta,\epsilon} \Phi_2(\alpha,\beta,\gamma,\delta,\epsilon)),
\end{eqnarray*}
where
\begin{eqnarray}\label{ttty}
\Phi_2(\alpha,\beta,\gamma,\delta,\epsilon)
&=&2(\alpha+\beta)\ln(\lambda)+H(\alpha)+H_1(\gamma,\alpha)+H_1(\alpha-\gamma,1-\alpha)+H(\beta)+H_1(\delta,\beta)\nonumber\\
&&+H_1(\beta-\delta,1-\beta) + \Delta\Big(H_1(\gamma,1-2\beta+\delta)-H(\gamma)+H_1(\epsilon,1-2\beta+\delta-\gamma)\nonumber\\
&&+H_1(\alpha-\gamma-\epsilon,\beta-\delta)-H_1(\alpha-\gamma,1-\gamma)+H_1(\alpha-\gamma,1-\beta-\gamma-\epsilon)\nonumber\\
&&-H_1(\alpha-\gamma,1-\alpha)\Big),\label{eq:phi2}
\end{eqnarray}
and $H(x)=-x\ln(x)-(1-x)\ln(1-x)$, $H_1(x,y)=-x(\ln(x)-\ln(y))+(x-y)(\ln(y-x)-\ln(y))$.

 To make $\Phi_2$ well defined, the variables have to satisfy $(\alpha,\beta) \in \mathcal{R}$ and
\begin{equation}\label{eq:con2}
\begin{split}
\gamma,\delta,\epsilon \geq 0, \quad \alpha-\gamma-\epsilon \geq 0, \quad \beta-\delta \geq 0,
\quad 1-2\beta+\delta-\gamma-\epsilon \geq 0,\\
1-\alpha-\beta-\epsilon\geq 0, \quad \beta-\delta+\epsilon+\gamma-\alpha\geq 0.
\end{split}
\end{equation}

Lemma \ref{lemam} tells us that in the non-uniqueness region (which is the
region of interest in this paper) the first moment is maximized when
$(\alpha,\beta)$ is $(p^+,p^-)$ (or symmetrically, $(p^-,p^+)$).
To show that these unbiased configurations dominate the Gibbs distribution with high
probability (as desired for Theorems \ref{thm:unbalanced-3} and \ref{thm:unbalanced-bigger}
we will apply the second moment method, as used in \cite{MWW}.

To that end, we need to analyze the second moment for
$(\alpha,\beta) = (p^+,p^-)$, and
show that $\E_{\G}[(Z^{\alpha,\beta}_G)^2]=O\big((\E_{\G}[Z^{\alpha,\beta}_G])^2\big)$.
To do that we need to show that for
for $(\alpha,\beta)=(p^+,p^-)$, $\E_{\G}[(Z^{\alpha,\beta}_G)^2]$
is roughly determined by uncorrelated pairs of configurations.
This crux of this is to show that $\Phi_2$ is maximized when
$\gamma=\alpha^2$ and $\delta=\beta^2$, which is detailed in the upcoming condition
which was proposed in~\cite{Sly10}.

\begin{condition}[Condition 1.2\footnote{The numbering in this paper for results from Sly's work \cite{Sly10} refer to the \texttt{arXiv} version of his paper.} of Sly \cite{Sly10}]
\label{condition:maxima}
There exists a constant $\chi>0$ such that when
$|p^+-\alpha|<\chi$ and $|p^--\beta|<\chi$ then
$g_{\alpha,\beta}(\gamma,\delta,\epsilon):=\Phi_2(\alpha,\beta,\gamma,\delta,\epsilon)$
achieves its unique maximum in the region~(\ref{eq:con2}) at the
point \[(\gamma^*,\delta^*,\epsilon^*)=(\alpha^2,\beta^2,\alpha(1-\alpha-\beta)).\]
\end{condition}

As mentioned earlier, Condition \ref{condition:maxima} implies
that for $(\alpha,\beta)=(p^+,p^-)$, $\E_{\G}[(Z^{\alpha,\beta}_G)^2]=O\big((\E_{\G}[Z^{\alpha,\beta}_G])^2\big)$.
While the implicit constant in the latter equality is bigger than one, an application of the small graph conditioning
method \cite{JLR,Wormald}
shows that for $(\alpha,\beta)=(p^+,p^-)$, $Z^{\alpha,\beta}_G$ is concentrated around  its expected value (up to a multiplicative arbitrarily small polynomial factor).  This yields a lower bound on the partition function $Z_G^{\alpha,\beta}$ for
 $(\alpha,\beta)=(p^+,p^-)$.  On the other hand, it is straightforward to an upper bound on the partition function for
 balanced configurations $\alpha=\beta$.  Consequently, one obtains Theorems
  \ref{thm:unbalanced-3} and \ref{thm:unbalanced-bigger} as detailed below in Section \ref{sec:contri}, which implies that
  the Gibbs distribution is unbalanced with high probability.
 Sly \cite{Sly10} uses random regular bipartite graphs as a gadget in his reduction and
 utilizes this bimodality property of the Gibbs distribution.

Before stating our new results on when Condition~\ref{condition:maxima} holds, it is useful to remind
the reader the previously known values of $\lambda$ for which Condition~\ref{condition:maxima} holds:
\begin{itemize}
\item $\Delta \geq 3$, and $\lambda_c(\Tree_\Delta)<\lambda<\lambda_c(\Tree_\Delta)+\varepsilon(\Delta)$ for some (small)  $\varepsilon(\Delta)>0$, (\cite[Lemma 6.10, Lemma 5.1]{MWW});
\item $\Delta=6$ and $\lambda=1$, (\cite[Section 5]{Sly10}).
\end{itemize}

Let $\lambda_{1/2}(\Tree_\Delta)$ be the smallest value
of $\lambda$ such that $\phi(\phi(1/2))=1/2$ ($\phi$ is the function defined in Lemma~\ref{lemam}). Equivalently  $\lambda_{1/2}(\Tree_\Delta)$ is the minimum solution of
\begin{equation}\label{eq:lambdahalf}
\left(1+(1/\lambda)^{1/\Delta}\right)^{1-1/\Delta}\left(1-(1/\lambda)^{1/\Delta}\right)^{1/\Delta}=1.
\end{equation}
The following Lemma is the technical core of this work.
\begin{lemma}\label{lem:cond}
Condition~\ref{condition:maxima} holds for
\begin{enumerate}
\item $\Delta=3$ and $\lambda > \lambda_c(\Tree_\Delta)$, and
\label{case:cond-3}
\item $\Delta>3$ and $\lambda\in (\lambda_c(\Tree_\Delta),\lambda_{1/2}(\Tree_\Delta)]$,
\label{case:cond-bigger}
\end{enumerate}
\end{lemma}

Lemma \ref{lem:cond} is proved in Section \ref{sec:analysis}.
As a corollary of Lemma \ref{lem:cond} we get that Condition \ref{condition:maxima}
holds for the range of $\lambda$ specified in Theorems  \ref{thm:unbalanced-3} and \ref{thm:unbalanced-bigger}.

\begin{corollary}
\label{cor:condition}
Condition~\ref{condition:maxima} holds for:
\begin{enumerate}
\item
\label{label-cor:three}
For $\Delta=3$ and $\lambda > \lambda_c(\Tree_3)$.
\item
\label{label-cor:six}
For $\Delta\geq 6$ and
$\lambda_c(\Tree_\Delta)<\lambda\leq \lambda_{1/2}(\Tree_{\Delta})$
and $\lambda_{1/2}(\Tree_{\Delta})>\lambda_c(\Tree_{\Delta-1})$.
  \item
  \label{label-cor:five}
For $\Delta=5$ and
$\lambda_c(\Tree_5)<\lambda\leq\lambda_c(\Tree_5) + .402912$.
    \item
    \label{label-cor:four}
For $\Delta=4$,
$\lambda_c(\Tree_4)<\lambda\leq\lambda_c(\Tree_4) + .327887$.
\end{enumerate}
\end{corollary}

\begin{proof}
Part \ref{label-cor:three} is identical to the first bullet in Lemma
\ref{lem:cond}.

Part \ref{label-cor:six} follows from the second bullet of Lemma
\ref{lem:cond} and the fact that for $\Delta\geq 6$, it holds $\lambda_{1/2}(\Tree_\Delta) > \lambda_c(\Tree_{\Delta-1})$. To see this, by~(\ref{eq:lambdahalf}), we have that $1-(1/\lambda_{1/2}(\Tree_\Delta))^{1/\Delta} > 0$, which implies that
$\lambda_{1/2}(\Tree_\Delta)>1$. For $\Delta \geq 6$, we have $\lambda_c(\Tree_\Delta)<1$. Hence, for $\Delta \geq 7$, we have
$\lambda_{1/2}(\Tree_\Delta) > \lambda_c(\Tree_{\Delta-1})$. For $\Delta=6$, the claim follows from the
fact that $\lambda_c(\Tree_5)=256/243<\lambda_{1/2}(\Tree_6)\approx 1.23105$.

For $\Delta=5$, note that
\[ \lambda_{1/2}(\Tree_5)-\lambda_c(\Tree_5)>1.45641-256/243> .402912,
\]
which proves Part \ref{label-cor:five}.

For $\Delta=4$, note that
\[ \lambda_{1/2}(\Tree_4)-\lambda_c(\Tree_4)>
2.015387 - 27/16 = .327887,
\]
which proves Part \ref{label-cor:four}.
\end{proof}

Theorems \ref{thm:unbalanced-3} and \ref{thm:unbalanced-bigger}
now follow from Corollary~\ref{cor:condition} as outlined earlier, and detailed
in the following subsection.

\subsection{Proofs of Main Theorems}\label{sec:contri}

We now proceed to prove Theorems~\ref{thm:unbalanced-3} and \ref{thm:unbalanced-bigger}.

\begin{proof}[Proofs of Theorems~\ref{thm:unbalanced-3} and \ref{thm:unbalanced-bigger}]
The proof is essentially the same as the proof of ~\cite[Theorem 2.2]{MWW} with minor modifications. We include the proof for the sake of completeness.

Choose $\delta>0$ such that for $$X:=\min_{|x-p^{+}|\leq \delta,\  |y-p^{-}| \leq\delta}\Phi_1(x,y),\ Y:=\max_{|x-y|\leq \delta}\Phi_1(x,y)$$ it holds that $\tau:=X-Y>0$.
To see that this is possible, note that $\Phi_1$ is continuous and hence it is uniformly continuous at any closed and bounded region. Since $\Phi_1$ exhibits a global maximum at $(p^{+},p^{-})$, the existence of $\delta$ follows.

By Markov's inequality, we obtain that a.a.s.
\begin{equation}\label{eq:upperbound}
\mu(\mathcal{I}^\delta_B)=\frac{\sum_{x,y:|x-y|\leq \delta} Z_{G}^{x,y}}{Z_G}\leq \frac{\exp(n(Y+\frac{\tau}{4}))}{Z_G}.
\end{equation}

To bound $\min\{\mu(\mathcal{I}^\delta_1),\mu(\mathcal{I}^\delta_2)\}$, we need the following result from~\cite{MWW}. While their result is only stated for $(\alpha,\beta)$ close to $(1/\Delta, 1/\Delta)$, it can readily be verified
(as Sly also observed, e.g., see the discussion before
Lemma 3.4 in \cite{Sly10})
that their proof holds in a neighbourhood of $(p^{+}, p^{-})$, whenever Condition~\ref{condition:maxima} holds.

\begin{lemma}[Theorem 3.4 of Mossel et al. \cite{MWW}]
\footnote{The stated version of the theorem differs slightly from the
version in \cite{MWW}.  In particular, in \cite{MWW} $(\alpha,\beta)$ is fixed,
in the sense that it's independent of $n$.  Here,
$(\alpha,\beta)$ depends on $n$.  In the application of this theorem in
the proof of Theorem 2.2 in \cite{MWW},
 it is unclear how they deduce the existence of a fixed $(\alpha,\beta)$
and this is why we modified the statement of the theorem.
Their proof of Theorem 3.4
in \cite{MWW} still goes through for this slightly modified version.}
\label{lem:MWWlowerbound}
Let $\Delta\geq 3$. Suppose that Condition~\ref{condition:maxima} holds,
then for all sufficiently large $n$ there exist $(\alpha_n,\beta_n)$ where
\begin{equation}
\label{eq:MWWcond}
  \alpha_n = p^+ + o(1), \ \
  \beta_n = p^- + o(1), \mbox{ and }
 n\alpha_n \mbox{ and } n\beta_n \mbox{ are integers},
\end{equation}
and it holds a.a.s. that $$Z^{\alpha_n,\beta_n}_G\geq \frac{1}{n}\E_{\G}[Z^{\alpha_n,\beta_n}_G].$$
\end{lemma}

  For all $n$ large enough, there exist $(\alpha_n, \beta_n)$
where $|\alpha_n - p^+| \leq 1/n$, $|\beta_n - p^-| \leq 1/n$ and
$n\alpha_n$ and $n\beta_n$ are integers, and therefore
\eqref{eq:MWWcond} holds.

 From Lemma~\ref{lem:MWWlowerbound} and Corollary~\ref{cor:condition} it follows that a.a.s. $Z_G^{\alpha,\beta}\geq \exp(n(X-\frac{\tau}{4}))$
and consequently $$\mu(\mathcal{I}^\delta_1)\geq \frac{\exp(n(X-\frac{\tau}{4}))}{Z_G}.$$
Since both $\Phi_1$ and $\Phi_2$ are symmetric with respect to $\alpha,\beta$, a similar statement to Lemma~\ref{lem:MWWlowerbound} holds with the roles of $p^{+},p^{-}$ interchanged, so that
\begin{equation}\label{eq:lowerbound}
\min\{\mu(\mathcal{I}^\delta_1),\mu(\mathcal{I}^\delta_2)\}\geq \frac{\exp(n(X-\frac{\tau}{4}))}{Z_G}.
\end{equation}
Combining \eqref{eq:upperbound} and \eqref{eq:lowerbound}, we obtain
\begin{equation}
\label{ineq:bad-cut}
\mu(\mathcal{I}^\delta_B)\leq \exp(-\frac{n\tau}{2})\min\{\mu(\mathcal{I}^\delta_1),\mu(\mathcal{I}^\delta_2)\}.
\end{equation}
This completes the proof.

The torpid mixing of the Glauber dynamics claimed in
Theorems \ref{thm:unbalanced-3} and \ref{thm:unbalanced-bigger}
follows from \eqref{ineq:bad-cut} 
by Claim 2.3 in \cite{DFJ},
which is a standard conductance argument.
\end{proof}

For Theorem \ref{thm:main}, we will use Lemma~\ref{lem:cond}
combined with the work of Sly~\cite{Sly10}, but we need one additional
ingredient.  The following combinatorial result will be used
to extend the inapproximability result for $\Delta=3$ to
a range of $\lambda$ for $\Delta\geq 6$.

\begin{lemma}\label{l:trans}
Let $G$ be a graph of maximum degree $\Delta$ and let $k>1$ be an
integer. Consider the graph $H$ obtained from $G$ be replacing
each vertex by $k$ copies of that vertex
and each edge by the complete bipartite
graph between the corresponding copies. Then,
$$Z_G( (1+\lambda)^k - 1)=Z_H(\lambda).$$
\end{lemma}

\begin{proof}
Consider the map $f:{\cal I}_H\rightarrow {\cal I}_G$ that maps an
independent set $I$ of $H$ to an independent set $J$ of $G$ such
that $v\in J$ if and only if at least one of the $k$ copies of $v$
in $H$ are in $I$.

For an independent set $J$ of $G$ the total contribution of sets
in $f^{-1}(J)$ to $Z_H(\lambda)$ is  $((1+\lambda)^k-1)^{|J|}$,
since for each $v\in J$ we can choose any non-empty subset of its
$k$ copies in $H$.
\end{proof}

We are now ready to prove Theorem~\ref{thm:main}.

\begin{proof}[Proof of Theorem~\ref{thm:main}]
Sly's reduction~\cite[Section 2]{Sly10}
establishes that for any $\lambda>\lambda_c(\Tree_\Delta)$ such that Condition~\ref{condition:maxima} holds and also the following two
inequalities hold:
\begin{equation}\label{eq:Slycond}
(\Delta-1)p^{+}p^{-}\leq (1-p^{+})(1-p^{-}) \mbox{ and } p^{+}<\frac{3}{5}(1-p^{-}),
\end{equation}
then there does not exist (assuming NP$\neq$RP) an FPRAS for the partition function of the
hard-core model with activity $\lambda$.
For the two inequalities in \eqref{eq:Slycond}, as
Sly \cite{Sly10} points out, they are unnecessary.
First off, the inequality $(\Delta-1)p^{+}p^{-}\leq (1-p^{+})(1-p^{-})$ is
implied by Part 3 of Lemma~\ref{lemam}.
And the condition $p^{+}<\frac{3}{5}(1-p^{-})$ is used to
simplify the proof of Lemma 4.2 in \cite{Sly10}. For completeness, a slight modification of Sly's Lemma 4.2 without \eqref{eq:Slycond} in the
hypothesis is proved in Section~\ref{sec:remainingSly}.

Corollary~\ref{cor:condition} establishes that
Condition~\ref{condition:maxima} holds
for a range of activities $\lambda$ for each $\Delta\geq 3$. By the discussion above, we obtain that there does not exist (assuming NP$\neq$RP) an FPRAS for the partition function of the hard-core model with activity
\begin{itemize}
\item $\lambda\in(\lambda_c(\Tree_3),\infty]$ for $\Delta=3$.
\item $\lambda\in(\lambda_c(\Tree_\Delta),\lambda_c(\Tree_{\Delta-1})]$ for $\Delta\geq 6$.
\item $\lambda\in(\lambda_c(\Tree_5),\lambda_c(\Tree_5) + .402912]$ for $\Delta=5$.
\item $\lambda\in(\lambda_c(\Tree_4),\lambda_c(\Tree_4) + .327887]$ for $\Delta=4$.
\end{itemize}

Lemma~\ref{l:trans} (used with $k=2$ and $\Delta=3$) implies that there does not exist (assuming
NP$\neq$RP) an FPRAS for the partition function of the hard-core model with activity $\lambda>\sqrt{5}-1$ in graphs of maximum
degree $6$.

The range of $\lambda$ for $\Delta$ and $\Delta-1$ where we can prove hardness (that is,
$(\lambda_c(\Tree_\Delta),\lambda_{c}(\Tree_{\Delta-1})]$) overlap for $\Delta\geq 6$.
This is useful since the hardness for $\Delta-1$ automatically gives hardness for $\Delta$.

Thus for $\Delta\geq 6$ we have the hardness result on the set
$$
(\lambda_c(\Tree_\Delta),\lambda_{c}(\Tree_{\Delta-1})]\cup\dots\cup
(\lambda_c(\Tree_6),\lambda_{c}(\Tree_5)]\cup (\lambda_{c}(\Tree_5), \lambda_{c}(\Tree_5)+0.402912]\cup (\sqrt{5}-1,\infty) =
(\lambda_c(\Tree_\Delta),\infty).
$$
This concludes the proof of the theorem for $\Delta\geq 6$.

For $\Delta=4$, we have the hardness result on the set $(\lambda_c(\Tree_4), 2.01538] \cup (4,+\infty)$.

For $\Delta=5$, we have the hardness result on the set $$(\lambda_c(\Tree_5), 1.45641] \cup (1.6875,2.01538] \cup (4,+\infty).$$
\end{proof}

\subsection{On the Use of Computational Assistance}
\label{sec:computer-assist}

We use Mathematica to prove several inequalities involving rational functions in regions
bounded by rational functions. Such inequalities are known to be decidable by Tarski's quantifier
elimination~\cite{MR0028796}, the particular version of Collins algebraic decomposition (CAD)
used by Mathematica's \verb!Resolve! command is described in~\cite{Strz}.
The algorithms are guaranteed to return correct answers---they do not suffer
from precision issues since they use interval arithmetic (a real number is
represented using an interval whose endpoints are rational numbers).

\section{Non-Reconstruction revisited}
\label{sec:remainingSly}

In this section, we reprove Lemma 4.2 from Sly \cite{Sly10} without
\eqref{eq:Slycond} in the hypothesis.  This will allow us to focus on proving Lemma~\ref{lem:cond}.

Recall, $\Tree_\Delta$ is the infinite $\Delta$-regular tree, and $p^+,p^-$
denote the marginal probabilities that the root is occupied for the limit of
even and odd, respectively, sized trees.
Let $\hat{\Tree}_\Delta$ denote the infinite $(\Delta-1)$-ary tree rooted at $\rho$.
(Thus, these two trees only differ at the root.)
 For $\ell\in\mathbb{N}$, let $\hat{T}_\ell$ denote
the tree with branching factor $\Delta-1$ and containing $\ell$ levels.
Let $q_\ell$ denote the marginal probability that the root is occupied
in the Gibbs distribution on $\hat{T}_\ell$.
Analogous to $p^+$ and $p^-$, let
  \[
 q^+ = \lim_{\ell\rightarrow\infty} q_{2\ell}
\ \ \mbox{ and } \ \
 q^- = \lim_{\ell\rightarrow\infty} q_{2\ell+1}.
\]
 The densities $q^{+},q^{-}$ are related to $p^{+},p^{-}$ by:
 $$q^{+}=\frac{p^{+}}{1-p^{-}}, \ \ \
 q^{-}=\frac{p^{-}}{1-p^{+}}.$$

 There are two semi-translation invariant measures $\hat{\mu}^+$ and $\hat{\mu}^-$ on $\hat{\Tree}_\Delta$
 obtained by taking the weak limit of the hard-core measure
 of even-sized trees $\hat{T}_{2\ell}$
 and odd-sized trees $\hat{T}_{2\ell+1}$, respectively.
 In these measures $\hat{\mu}^+, \hat{\mu}^-$ on the infinite tree, $q^+$ and
 $q^-$, respectively, are the marginal probabilities that the root is occupied.
 These measures  $\hat{\mu}^+$ and $\hat{\mu}^-$ can also be
 generated by a broadcasting process, see \cite[Section 4]{Sly10}.

 For $v\in\hat{\Tree}_\Delta$, denote by $S_{v,\ell}$ the vertices at level $\ell$ in the subtree of $\hat{\Tree}_\Delta$ rooted at $v$.
 Let $X_{\rho,\ell,+}$ denote the marginal probability that the root $\rho$ is
 occupied in an independent set $X$ generated by the following process:
 we first sample an independent set $\hat{X}$ from the measure $\hat{\mu}^+$,
 then we condition on the configuration $\hat{X}_S$ on $S_{\rho,\ell}$,
 and finally we sample an independent set $X$ from
 the hard-core measure on $\hat{T}_\ell$ conditioned on $X_S=\hat{X}_S$.
 Note that the configuration of the vertices in $S_{\rho,\ell}$ is a random vector, so that $X_{\rho,\ell,+}$ is a random variable. Define similarly $X_{\rho,\ell,-}$ for $\hat{\mu}^-$.

  We may extend this definition to an arbitrary vertex $v\in\hat{\Tree}_\Delta$ at distance $D$ from the root $\rho$, by setting $X_{v,\ell,+}=X_{\rho, \ell, +}$ if $D$ is even or $X_{v,\ell,+}=X_{\rho, \ell, -}$ if $D$ is odd. Thus, $X_{v,\ell,+}$ is the probability that, in an appropriate translation of $\hat{\mu}^+$, $v$ is occupied conditioning on the configuration of $S_{v,\ell}$. Define similarly $X_{v,\ell,-}$ for $\hat{\mu}^-$.

In Sly~\cite{Sly10}, it is proved that $X_{\rho,\ell,+}$ ($X_{\rho,\ell,-}$) is strongly concentrated around $q^+$ ($q^-$, respectively)
under the condition that $p^+<\frac{3}{5}(1-p^-)$. This concentration is used
by Sly for establishing the properties of the gadget he uses in his reduction.
We need to reprove the concentration without the condition $p^+<\frac{3}{5}(1-p^-)$.

\begin{lemma}[see Lemma 4.2 in Sly~\cite{Sly10}]
\label{lem:Sly42}
When $\lambda>\lambda_c(\Tree_\Delta)$, there exist constants $\zeta_1(\lambda, \Delta), \zeta_2(\lambda, \Delta) > 0$, for all sufficiently large $\ell$,
$$\mathbb{P}\left[\Big\vert X_{\rho,\ell,+} - q^+\Big\vert  \geq \exp(-\zeta_1\ell)\right]\leq \exp(-\exp(\zeta_2l)).$$
$$\mathbb{P}\left[\Big\vert X_{\rho,\ell,-} - q^-\Big\vert  \geq \exp(-\zeta_1\ell)\right]\leq \exp(-\exp(\zeta_2l)).$$
\end{lemma}

\begin{proof}
We follow the proof of~\cite{Sly10}.  For a vertex $v\in\hat{\Tree}_\Delta$, let
$N(v)=N^1(v)$ denote its children, and for $i\in\mathbb{N}$, let
$N^i(v)$ denote its descendants $i$ levels below.
For $1\le i\le \ell$ and $s\in\{+,-\}$, let $\mathcal{X}_{N^i(v),\ell-i,s}$
  denote the vector $\{X_{w,\ell-i,s}: w\in N^i(v)\}$.

By standard tree recursions, it can be proved that for $s\in\{+,-\}$,
\begin{equation*}
X_{v,\ell,s}=h(\mathcal{X}_{N(v),\ell-1,s})
   :=\frac{\lambda\prod_{w\in N(v)}(1-X_{w,\ell-1,s})}{1+\lambda\prod_{w\in N(v)}(1-X_{w,\ell-1,s})}.
\end{equation*}
In our framework, one may establish a contraction property for $h$, but it is slightly more straightforward to look at depth two of the recursion. Note that
\begin{equation*}
1-X_{v,\ell,s}=\frac{1}{1+\lambda\prod_{w\in N(v)}(1-X_{w,\ell-1,s})}
\end{equation*}
so that
\begin{align}
X_{v,\ell,s}&=1-\frac{1}{1+\lambda\prod_{w\in N(v)}\frac{1}{1+\lambda\prod_{z\in N(w)}(1-X_{z,\ell-2,s})}}\notag\\
&=\frac{\lambda}{\lambda+\prod_{w\in N(v)}\left(1+\lambda\prod_{z\in N(w)}(1-X_{z,\ell-2,s})\right)} \label{eq:recur}
\\
\nonumber
&=:r(\mathcal{X}_{N^{2}(\rho),\ell-2,s}).
\end{align}
By recursively applying \eqref{eq:recur}, we obtain that \[X_{\rho,\ell,s}=:g(\mathcal{X}_{N^{2L}(\rho),\ell-2L,s}).\]
Note that $r(\mathcal{X}_{N^{2}(\rho),\ell-2,s})=g(\mathcal{X}_{N^{2}(\rho),\ell-2,s})$.

For the rest of the proof, we focus on the case $s=+$ and $v$ being at even distance from the root~$\rho$, the other cases being almost identical (since we are looking at depth two of the recursion).

Assume $\mathcal{X}_{N^2(v),\ell-2,+},\mathcal{X}'_{N^2(v),\ell-2,+}$ are two vectors which are equal except at one vertex $z^*{\in}N^2(v)$.
To prove the new version of the lemma, it suffices to obtain the
following contraction property.
We will prove for all $\ell$ sufficiently large, for all $v\in \hat{\Tree}_\Delta$,
for all
pairs $\mathcal{X}_{N^2(v),\ell-2,+},\mathcal{X}'_{N^2(v),\ell-2,+}$
that differ at a single vertex $z^*$,
\begin{equation}\label{eq:contraction}
|r(\mathcal{X}_{N^2(v),\ell-2,+})-r(\mathcal{X}'_{N^2(v),\ell-2,+})|\leq \left(\frac{1}{(\Delta-1)^2}+\frac{1}{10}\right) |X_{z^*,\ell-2,+}-X'_{z^*,\ell-2,+}|.
\end{equation}
To accomplish this, it suffices to prove that for $\ell$ large enough, it holds that:
\begin{equation}\label{eq:contfactor}
\left|\frac{\partial r}{\partial X_{z^*,\ell-2,+}}(\mathcal{X}_{N^2(v),\ell-2,+})\right|<\frac{1}{(\Delta-1)^2}+\frac{1}{10}.
\end{equation}
Let $w^*\in N(v)$ be the parent of $z^*$.  Then \eqref{eq:contfactor}
is equivalent to:
\begin{equation*}
\frac{\lambda^2\prod_{z\in N(w^*)\setminus\{z^*\}}(1-X_{z,\ell-2,+})
 \prod_{w\in N(v)\setminus\{w^*\}}\left(1+\lambda\prod_{z\in N(w)}(1-X_{z,\ell-2,+})\right)}{\left[\lambda+\prod_{w\in N(v)}\left(1+\lambda\prod_{z\in N(w)}(1-X_{z,\ell-2,+})\right)\right]^2}<\frac{1}{(\Delta-1)^2}+\frac{1}{10}.
\end{equation*}

Note that as $\ell\rightarrow+\infty$, since $v$ is at even distance from $\rho$, $X_{v,\ell,+}$ converges almost surely to $q^{+}$ (see \cite[equation (4.1)]{Sly10} and \cite[proof of Lemma 4.2]{Sly10}).
Hence, as $\ell\rightarrow+\infty$, we have that
\begin{equation*}
\left|\frac{\partial r}{\partial X_{z^*,\ell-2,+}}(\mathcal{X}_{N^2,\ell-2,+})\right|\rightarrow \frac{\lambda^2(1-q^+)^{\Delta-2}\left(1+\lambda(1-q^+)^{\Delta-1}\right)^{\Delta-2}}{\left[\lambda+\left(1+\lambda(1-q^+)^{\Delta-1}\right)^{\Delta-1}\right]^2}=: \gamma.
\end{equation*}
Pick $\ell_0=\ell_0(\Delta,\lambda)$ so that for all $\ell\geq\ell_0$ we have that \[\left|\frac{\partial r}{\partial X_{z^*,\ell-2,+}}(\mathcal{X}_{N^2(v),\ell-2,+})-\gamma\right|\leq \frac{1}{10}.\] We will see next that $\gamma=q^{+}q^{-}$ and hence by Part 3 of Lemma~\ref{lemam}, it follows that $\gamma\leq \frac{1}{(\Delta-1)^2}$ which proves \eqref{eq:contfactor}.

To see that $\gamma=q^+q^-$, note that $\lambda (1-q^+)^{\Delta-1}=\frac{q^-}{1-q^-}$ and $\lambda (1-q^-)^{\Delta-1}=\frac{q^+}{1-q^+}$ (one can derive these equalities by $p^+=\phi(p^-), p^-=\phi(p^+)$, see Part 1 of Lemma~\ref{lemam}). It is now a matter of a few algebra substitutions to check that $\gamma=q^+q^-$.

%% We claim\footnote{This is the point where we deviate slightly from
%% the proof of Lemma 4.2 in \cite{Sly10}.  In his proof, he uses the
%%extra condition $p^+<\frac{3}{5}(1-p^-)$ to deduce:
%%$|g(\mathcal{X}_{1,\ell,s})-g(\mathcal{X}'_{1,\ell,s})|\leq \frac{3}{4}|X_{u,\ell-1,s}-X'_{u,\ell-1,s}|.$}
This proves \eqref{eq:contraction}.
Recursively applying this relation \eqref{eq:contraction} implies
that for all $\ell-L>\ell_0$, for $\Delta\geq 3$,
\begin{equation}\label{eq:cont-rec}
|g(\mathcal{X}_{N^L(v),\ell-L,+})-g(\mathcal{X}'_{N^L(v),\ell-L,+})|\leq
\left(\frac{1}{2}\right)^{\lfloor L/2 \rfloor} |X_{z^*,\ell-L,+}-X'_{z^*,\ell-L,+}|.
\end{equation}
Having established \eqref{eq:cont-rec},
the rest of Sly's proof of Lemma 4.2 goes through.
\end{proof}

\section{Analysis of the Second Moment}\label{sec:analysis}
In this section, we prove Lemma \ref{lem:cond}. We break its proof into two slightly easier components, namely
Lemmata~\ref{lem:main} and~\ref{lem:d3} stated below, depending on the value of $p^+$. Notice that Lemma~\ref{lem:main} determines explicitly a region of $(\alpha,\beta)$ for which Condition~\ref{condition:maxima} holds.

\begin{lemma}\label{lem:main}
Let $\Delta \geq 3$ and let $(\alpha,\beta) \in \TD$, $\alpha,\beta>0$, and $\alpha,\beta \leq
1/2$. Then $g_{\alpha,\beta}(\gamma,\delta,\epsilon):=\Phi_2(\alpha,\beta,\gamma,\delta,\epsilon)$
achieves its unique maximum in the region~(\ref{eq:con2}) at the point $(\gamma^*,\delta^*,\epsilon^*)$.
\end{lemma}

\begin{lemma}\label{lem:d3}
Fix $\Delta=3$ and $\lambda>\lambda_c(\Tree_\Delta)$. Let $p^+$ and $p^-$ be the corresponding
probabilities. Assume that $1/2 \leq p^+ < 1$. There exists a constant $\chi>0$ such that for
$|p^+-\alpha|<\chi$ and $|p^--\beta|<\chi$,
$g_{\alpha,\beta}(\gamma,\delta,\epsilon):=\Phi_2(\alpha,\beta,\gamma,\delta,\epsilon)$
achieves its unique maximum in the region~(\ref{eq:con2}) at the
point $(\gamma^*,\delta^*,\epsilon^*)$.
\end{lemma}

\begin{proof}[Proof of Lemma \ref{lem:cond}]
We begin by proving the second part of the lemma ($\Delta>3$ case) and  then we prove the first part ($\Delta=3$ case).

Fix any $\Delta\geq 3$ and $\lambda>\lambda_c(\Tree_\Delta)$. Let $p^+,\ p^-$ be
the corresponding marginal probabilities that the root is occupied in the measures $\mu^+,\ \mu^-$. By the third item of Lemma~\ref{lemam}, we have that $(\alpha,\beta)=(p^+,p^-)$
is contained in the interior of $\TD$. Hence, the same holds for every $(\alpha,\beta)$ in a small enough neighborhood of $(p^+,p^-)$. Hence, provided that $p^+\leq 1/2$, Lemma~\ref{lem:main} verifies the condition for any such point.  It is now easy to see that $p^+$ is increasing in $\lambda$, and hence $p^+ \leq 1/2$ iff $\lambda\in (\lambda_c(\Tree_\Delta),\lambda_{1/2}(\Tree_\Delta)]$. This proves the second part of Lemma~\ref{lem:cond}.

The first part of Lemma~\ref{lem:cond} is proved analogously. When $p^+\leq 1/2$, one uses Lemma~\ref{lem:main} as above, while when $p^+\geq 1/2$ the condition reduces to Lemma~\ref{lem:d3}.
\end{proof}

Thus we may focus our attention on proving Lemmata~\ref{lem:main} and~\ref{lem:d3}. While the analysis at some point for Lemma~\ref{lem:d3} requires tighter arguments, the two proofs share many common preprocessing steps. The rest of this section is devoted to these common preprocessing steps and the proofs of Lemmata~\ref{lem:main} and~\ref{lem:d3} are given in Section~\ref{sec:conclu}.

\subsection{The Partial Derivatives}
The derivatives of $\Phi_2$ with respect to $\gamma,\delta,\epsilon$ can easily be computed and are also given in \cite[Proof of Lemma 6.1]{MWW}:
\begin{eqnarray}
\exp\left(\frac{\partial\Phi_2}{\partial\gamma}\right) &=& \frac{(1-2\beta+\delta-\gamma-\epsilon)^{\Delta}(\alpha-\gamma-\epsilon)^{\Delta}(1-2\alpha+\gamma)^{\Delta-1}}{(1-\beta-\gamma-\epsilon)^{\Delta}(\beta-\alpha+\gamma-\delta+\epsilon)^{\Delta}(\alpha-\gamma)^{\Delta-2}\gamma},\label{eq:dphigm}\\
\exp\left(\frac{\partial\Phi_2}{\partial\delta}\right) &=& \frac{(\beta-\alpha-\delta+\gamma+\epsilon)^{\Delta}(1-2\beta+\delta)^{\Delta-1}}{(1-2\beta+\delta-\gamma-\epsilon)^{\Delta}(\beta-\delta)^{\Delta-2}\delta},\label{eq:dphidt}\\
\exp\left(\frac{\partial\Phi_2}{\partial\epsilon}\right) &=& \frac{(1-2\beta+\delta-\gamma-\epsilon)^{\Delta}(\alpha-\gamma-\epsilon)^{\Delta}(1-\alpha-\beta-\epsilon)^{\Delta}}{\epsilon^{\Delta}(\beta-\alpha-\delta+\gamma+\epsilon)^{\Delta}(1-\beta-\gamma-\epsilon)^{\Delta}}.\label{eq:dphiel}
\end{eqnarray}

\subsection{Excluding the Boundary of the Region}\label{sec:boundex}
We argue that for $(\alpha,\beta)\in \TD$,  the maximum of $g_{\alpha,\beta}$ cannot occur on
the boundary of the region defined by~(\ref{eq:con2}).

\begin{lemma}\label{lem:interior}
For every $\Delta \geq 3$ and $(\alpha,\beta) \in \TD$,
$g_{\alpha,\beta}(\gamma,\delta,\epsilon):=\Phi_2(\alpha,\beta,\gamma,\delta,\epsilon)$ attains
its maximum in the interior of~(\ref{eq:con2}).
\end{lemma}

\begin{proof}
We follow the proof of Lemma 6.1 of~\cite{MWW} (where the same result is proved in the special case $\alpha=\beta=\frac{1}{\Delta}$).

We will prove that $g_{\alpha,\beta}(\gamma,\delta,\epsilon)$ attains
its maximum in the interior of~(\ref{eq:con2}) by showing that at least
one derivative in~(\ref{eq:dphigm})--(\ref{eq:dphiel}) goes to infinity
($+$ or $-$ according to the direction),
as we approach one of the boundaries defined by~(\ref{eq:con2}) from the
interior of~(\ref{eq:con2}).

For $\Delta\geq 3$, note that $(\alpha,\beta) \in \TD$ implies $\alpha+\beta<1$.
Without loss of generality,
we assume that $\alpha\geq \beta$. Hence $\beta < 1/2$.
We have~(\ref{eq:dphidt}) goes
to $+\infty$ as $\delta \to 0$, and~(\ref{eq:dphiel}) goes
to $+\infty$ as $\epsilon \to 0$. We also have~(\ref{eq:dphiel}) goes to
$-\infty$ as $\gamma+\epsilon \to \alpha$, (\ref{eq:dphidt}) goes to
$-\infty$ as $\delta \to \beta$, (\ref{eq:dphiel}) goes to $-\infty$ as
$\gamma+\epsilon-\delta \to 1-2\beta$, (\ref{eq:dphiel}) goes to $-\infty$ as
$\epsilon \to 1-\alpha-\beta$, and~(\ref{eq:dphiel}) goes to $+\infty$ as
$\gamma+\epsilon-\delta \to \alpha-\beta$.

When $\alpha \geq 1/2$, the condition $\gamma = 0$ is not a boundary,
as $\gamma \geq 0$ is implied by the conditions $\delta \geq 0$,
$1-\alpha-\beta-\epsilon\geq 0$, and $\beta-\delta+\epsilon+\gamma-\alpha\geq 0$.
On the other hand, when $\alpha < 1/2$, we have~(\ref{eq:dphigm}) goes to $+\infty$ as $\gamma \to 0$.
\end{proof}

\subsection{Eliminating one variable}\label{sec:variableel}

Fix $\Delta,\alpha,\beta,\gamma,\delta$ and view $\Phi_2$ as a function of $\epsilon$. We maximize with respect to $\epsilon$. In this setting, it was proved ~\cite[Lemma 6.3]{MWW}, that the only maximizer of the function $\Phi_2$ in the interior of~(\ref{eq:con2}) is obtained by solving $\frac{\partial\Phi_2}{\partial\epsilon}=0$ and is given by:
\begin{equation}\label{eq:solel}
\hat{\epsilon}:=\hat{\epsilon}(\alpha,\beta,\gamma,\delta) = \frac{1}{2}(1+\alpha-\beta-2\gamma - \sqrt{D}),
\end{equation}
where
\begin{eqnarray*}
D &=& (1+\alpha-\beta-2\gamma)^2-4(\alpha-\gamma)(1-2\beta-\gamma+\delta)\\
&=& (\alpha+\beta-1)^2+4(\alpha-\gamma)(\beta-\delta).
\end{eqnarray*}

Define
\begin{equation}\label{eq:solet}
\hat{\eta}:=\hat{\eta}(\alpha,\beta,\gamma,\delta) = \frac{1}{2}(1-\alpha+\beta-2\delta - \sqrt{D}),
\end{equation}
and note that
\begin{equation}\label{eq:usefulineq1}
\begin{split}
\alpha-\gamma-\hat{\epsilon}=\beta-\delta-\hat{\eta}=\frac{1}{2}(-(1-\alpha-\beta)+\sqrt{D}),\\ (\alpha-\gamma-\hat{\epsilon})(1-\alpha-\beta-\hat{\epsilon}-\hat{\eta})=\hat{\epsilon}\hat{\eta}.
\end{split}
\end{equation}
The new parameter~$\hat{\eta}$ (not used in~\cite{MWW}) is symmetric with $\hat{\epsilon}$, i.e, the constraints and formulas we
have are invariant under a symmetry that swaps $\alpha,\gamma,\hat{\epsilon}$ with $\beta,\delta,\hat{\eta}$. This will allow for simpler arguments (using the symmetry).

From the previous discussion and equation~\eqref{eq:solel}, we may eliminate variable $\epsilon$ of our consideration. Of course, this introduces some complexity due to the radical $\sqrt{D}$, but still this is manageable.  Let
$$
f(\gamma,\delta):=g_{\alpha,\beta}(\gamma,\delta,\hat{\epsilon})=\Phi_2(\alpha,\beta,\gamma,\delta,\hat{\epsilon}).
$$
To prove that $(\gamma^*,\delta^*,\epsilon^*)$ is the unique global maximum of $g_{\alpha,\beta}$
in the interior of the region defined by
(\ref{eq:con2}), it suffices to prove that $(\gamma^*,\delta^*)$ is the unique global maximum of $f$ for $(\gamma,\delta)$ in
the interior of the following region, which contains the $(\gamma,\delta)$-projection of the region
 defined by (\ref{eq:con2}):
\begin{equation}\label{eq:con3}
%\begin{split}
0 \leq \gamma \leq \alpha, \quad 0 \leq \delta \leq \beta,
\quad
%\\
0 \leq 1-2\beta+\delta-\gamma, \quad 0 \leq 1-2\alpha+\gamma-\delta.
%\end{split}
\end{equation}
Each inequality in~\eqref{eq:con3} is implied by the inequalities in (\ref{eq:con2}), the only
non-trivial case being the last inequality which is the sum of $1-\alpha-\beta-\epsilon\geq 0$
and $\beta-\delta+\epsilon+\gamma-\alpha\geq 0$.

The first derivatives of $f$ with respect to $\gamma,\delta$ are
\begin{eqnarray}
\frac{\partial f}{\partial \gamma}(\gamma,\delta) &=& \Delta\ln{W_{11}}+\ln{W_{12}},\label{eq:firstd1}\\
\frac{\partial f}{\partial \delta}(\gamma,\delta) &=& \Delta\ln{W_{21}}+\ln{W_{22}},\label{eq:firstd2}
\end{eqnarray}
where
\begin{equation*}
\begin{split}
W_{11} = \frac{(\alpha-\gamma-\hat{\epsilon})\hat{\epsilon}(1-2\alpha+\gamma)}{\hat{\eta}(\alpha-\gamma)^{2}}=\frac{\hat{\epsilon}(1-2\alpha+\gamma)}{(1-\alpha-\beta-\hat{\epsilon})(\alpha-\gamma)},\quad
W_{12} = \frac{(\alpha-\gamma)^2}{(1-2\alpha+\gamma)\gamma},\\
W_{21} = \frac{(\beta-\delta-\hat{\eta})\hat{\eta}(1-2\beta+\delta)}{\hat{\epsilon}(\beta-\delta)^{2}}=\frac{\hat{\eta}(1-2\beta+\delta)}{(1-\alpha-\beta-\hat{\eta})(\beta-\delta)},\quad
W_{22} = \frac{(\beta-\delta)^2}{(1-2\beta+\delta)\delta}.
\end{split}
\end{equation*}
Note that the rightmost equalities in the definition of $W_{11}$ and $W_{21}$ follow
from~(\ref{eq:solet}) and~(\ref{eq:usefulineq1}).

For every $\Delta \geq 3$, and $(\alpha,\beta) \in \TD$, we have that $(\gamma^*,\delta^*)$ is a stationary point of $f$
(this follows from the fact that for $\gamma=\alpha^2$ and $\delta=\beta^2$, the inequalities on the right-hand sides
in Lemma~\ref{lem:zeros} become equalities, and from \eqref{eq:firstd1}, \eqref{eq:firstd2} we have that
the derivatives of $f$ vanish).

\subsection{Restricting the Region}\label{sec:regionre}

To determine whether~\eqref{eq:firstd1} and~\eqref{eq:firstd2} are zero it will
be useful to understand conditions that make $W_{11},W_{12},W_{21},W_{22}$
greater or equal to $1$. The following lemma gives such conditions. The proof is given in Section~\ref{sec:lemzeros}.
\begin{lemma}\label{lem:zeros}
For every $(\alpha,\beta) \in \mathcal{R}$, and $(\gamma,\delta)$ in the interior of~(\ref{eq:con3}),
\begin{eqnarray*}
W_{11} \geq 1 &\iff& (1-\alpha)^2\delta+\beta^2(2\alpha-1-\gamma) \geq 0,\\
W_{12} \geq 1 &\iff& \gamma \leq \alpha^2,\\
W_{21} \geq 1 &\iff& (1-\beta)^2\gamma+\alpha^2(2\beta-1-\delta) \geq 0,\\
W_{22} \geq 1 &\iff& \delta \leq \beta^2.
\end{eqnarray*}
\end{lemma}

%The proof of Lemma~\ref{lem:zeros} is deferred to Section~\ref{sec:proofs}.

By considering the sign of \eqref{eq:firstd1} and \eqref{eq:firstd2} and Lemma~\ref{lem:zeros},
we have that the stationary points of $f$ (and hence of $g_{\alpha,\beta}$)
 can only be in
\begin{equation}\label{eq:regionlower}
%\begin{split}
0< \gamma \leq \alpha^2,\  0< \delta \leq \beta^2, \
%\\
(1-\alpha)^2\delta+\beta^2(2\alpha-1-\gamma) \leq 0,\    (1-\beta)^2\gamma+\alpha^2(2\beta-1-\delta) \leq 0,
%\end{split}
\end{equation}
or
\begin{equation}\label{eq:regionupper}
%\begin{split}
\alpha^2 \leq \gamma < \alpha,\ \beta^2 \leq \delta < \beta, \
%\\
(1-\alpha)^2\delta+\beta^2(2\alpha-1-\gamma) \geq 0,\   (1-\beta)^2\gamma+\alpha^2(2\beta-1-\delta) \geq 0.
%\end{split}
\end{equation}
Note that for $\gamma=\alpha^2, \delta=\beta^2$ one has $W_{11}=W_{12}=W_{21}=W_{22}=1$ (Lemma~\ref{lem:zeros} holds with equalities as well instead of inequalities), so that $(\alpha^2,\beta^2)$ is always a stationary point for $f(\gamma,\delta)$.

\subsection{The Hessian}\label{sec:hessiancons}

To prove that $f$ has a unique maximum, we are going to argue that $f$ is  strictly
 concave in each of the regions defined by~\eqref{eq:regionlower} and~\eqref{eq:regionupper}.
 It will thus be crucial to study the Hessian of $f$.

 Let $\mathbf{H}$ denote the Hessian of $f$, i.e.,
$$
\mathbf{H}=
\begin{pmatrix}
\frac{\partial f}{\partial^2 \gamma}(\gamma,\delta) & \frac{\partial f}{\partial\gamma \partial\delta}(\gamma,\delta)\\
\frac{\partial f}{\partial\delta \partial\gamma}(\gamma,\delta) &
\frac{\partial f}{\partial^2 \delta}(\gamma,\delta)
\end{pmatrix}.
$$

Our goal is to express $\mathbf{H}$ in a helpful explicit form. In this vein, it will be convenient to define the following quantities.
\begin{equation*}
\begin{array}{ccc}
R_1 = \frac{1-\alpha-\beta}{1-\alpha-\beta-\hat{\epsilon}-\hat{\eta}}, & R_2 = \frac{\sqrt{D}}{1-2\alpha+\gamma}, & R_3 = \frac{2(\alpha-\gamma-\hat{\epsilon})}{\alpha-\gamma},\\
R_4 = \frac{\sqrt{D}}{\gamma}, & R_5 = \frac{2(1-\beta-\gamma-\hat{\epsilon})}{\alpha-\gamma}, & R_6 = \frac{\sqrt{D}}{1-2\beta+\delta},\\
R_7 = \frac{2(\alpha-\gamma-\hat{\epsilon})}{\beta-\delta}, & R_8 = \frac{\sqrt{D}}{\delta}, & R_9 = \frac{2(1-\beta-\gamma-\hat{\epsilon})}{\beta-\delta}.
\end{array}
\end{equation*}
$\mathbf{H}$ can now be written in a relatively nice form with respect to the $R_i$. Namely,
\begin{eqnarray}
\frac{\partial f}{\partial^2 \gamma}(\gamma,\delta)
&=& \frac{1}{\sqrt{D}}\Big[(-R_1+R_2+R_3)\Delta-R_2-R_3-R_4-R_5\Big],\label{eq:secondd1}\\
\frac{\partial f}{\partial^2 \delta}(\gamma,\delta)
&=& \frac{1}{\sqrt{D}}\Big[(-R_1+R_6+R_7)\Delta-R_6-R_7-R_8-R_9\Big],\label{eq:secondd2}\\
\frac{\partial f}{\partial\gamma \partial\delta}(\gamma,\delta) &=& \frac{\partial f}{\partial\delta \partial\gamma}(\gamma,\delta) = \frac{\Delta R_1}{\sqrt{D}}.\label{eq:secondd3}
\end{eqnarray}

Inspecting the $R_i$ we obtain the following observation.
\begin{observation}\label{cpositive}
$R_1,\dots,R_9$ are positive when $(\alpha,\beta) \in \TD$ and $(\gamma,\delta)$ in the interior of~(\ref{eq:con3}).
\end{observation}
Observation~\ref{cpositive} and equation \eqref{eq:secondd3} immediately yield:
\begin{observation}\label{ccc2}
For every $(\alpha,\beta) \in \TD$, and $(\gamma,\delta)$ in the interior of~(\ref{eq:con3}),
$$
\frac{\partial f}{\partial\gamma \partial\delta}(\gamma,\delta)=\frac{\partial f}{\partial\delta \partial\gamma}(\gamma,\delta)>0.
$$
\end{observation}

In Section~\ref{sec:proof-secondd} we prove the following technical inequality on the $R_i$.
\begin{lemma}\label{lem:secondd}
For every $(\alpha,\beta) \in \mathcal{R}$, and $(\gamma,\delta)$ in the interior of~(\ref{eq:con3}),
$$
R_1>R_2+R_3, \quad\mbox{and}\quad R_1>R_6+R_7.
$$
\end{lemma}

Applying Lemma~\ref{lem:secondd} and Observation~\ref{cpositive}, \eqref{eq:secondd1} and \eqref{eq:secondd2} give the following straightforward corollary.
\begin{corollary}\label{ccc1}
For every $(\alpha,\beta) \in \mathcal{R}$, and $(\gamma,\delta)$ in the interior of~(\ref{eq:con3}),
$$
\frac{\partial f}{\partial^2 \gamma}(\gamma,\delta)<0, \quad \frac{\partial f}{\partial^2 \delta}(\gamma,\delta)<0.
$$
\end{corollary}

Corollary \ref{ccc1} implies that the sum of the eigenvalues of $\mathbf{M}$ is negative. Hence, $\mathbf{H}$ is negative definite iff the determinant of $\mathbf{H}$ is negative. We have the following expression for $\det(\mathbf{H})$.
\begin{eqnarray}
\det(\mathbf{H})
&=& \frac{\partial f}{\partial^2 \gamma}(\gamma,\delta) \cdot \frac{\partial f}{\partial^2 \delta}(\gamma,\delta) - \frac{\partial f}{\partial\gamma \partial\delta}(\gamma,\delta) \cdot \frac{\partial f}{\partial\delta \partial\gamma}(\gamma,\delta)\nonumber\\
&=& \frac{1}{D}\bigg\{\ (\Delta-1)^2\Big[(-R_1+R_2+R_3)(-R_1+R_6+R_7)-R_1^2\Big]\bigg.\nonumber\\
&&\ \ \ \ \ + (\Delta-1)\Big[ (-R_1+R_2+R_3)(-R_1-R_8-R_9)+ \Big. \nonumber \\
& & \hskip 3cm \Big. +  (-R_1+R_6+R_7)(-R_1-R_4-R_5)-2R_1^2\Big] \nonumber\\
&& \ \ \ \ \ + \bigg.\Big[(-R_1-R_8-R_9)(-R_1-R_4-R_5)-R_1^2\Big]\bigg\}.\label{eq:det}
\end{eqnarray}

\section{Concluding the Proofs of Lemmata~\ref{lem:main} and~\ref{lem:d3}}\label{sec:conclu}
In this section, we give the proofs of Lemmata~\ref{lem:main} and~\ref{lem:d3}. We first recap what we have accomplished in Section~\ref{sec:analysis} for every $(\alpha,\beta)\in\TD$.
\begin{enumerate}
\item The function $g_{\alpha,\beta}(\gamma,\delta,\epsilon)$ attains its maximum in the interior of the region \eqref{eq:con2}. See Section~\ref{sec:boundex}.
\item To study the (local) maxima of $g$ in the interior of the region \eqref{eq:con2}, it suffices to study the maxima of the function $f(\gamma,\delta)=g_{\alpha,\beta}(\gamma,\delta,\hat{\epsilon})$ in the interior of the region \eqref{eq:con3}. More explicitly, if $f(\gamma,\delta)$ has a unique maximum in the interior of the region \eqref{eq:con3} at $(\gamma^*,\delta^*)=(\alpha^2,\beta^2)$, then $g_{\alpha,\beta}$ has a unique  maximum in the interior of the region \eqref{eq:con2} at $(\gamma^*,\delta^*,\epsilon^*)=\big(\alpha^2,\beta^2,\alpha(1-\alpha-\beta)\big)$. The function $f$ is differentiable in the interior of the region \eqref{eq:con3}. See Section~\ref{sec:variableel}.
\item The point $(\gamma^*,\delta^*)=(\alpha^2,\beta^2)$ is always stationary for $f$. Every stationary point of $f$ lies in one of the two regions defined by (i) \eqref{eq:con3} and \eqref{eq:regionlower}, (ii) \eqref{eq:con3} and \eqref{eq:regionupper}. See Section~\ref{sec:regionre}.
\item The function $f$ is strictly concave iff $\det(\mathbf{H})>0$. See Section~\ref{sec:hessiancons}.
\end{enumerate}

By the above discussion, if for some $(\alpha,\beta)\in\TD$ it holds that $f$ is strictly concave in each of the regions (i) \eqref{eq:con3} and \eqref{eq:regionlower}, (ii) \eqref{eq:con3} and \eqref{eq:regionupper}, then $g_{\alpha,\beta}$ has a unique maximum at $(\gamma^*,\delta^*,\epsilon^*)=\big(\alpha^2,\beta^2,\alpha(1-\alpha-\beta)\big)$. Thus, it suffices to check that $\det(\mathbf{H})>0$ in each of these regions. This is essentially the way we derive Lemmata~\ref{lem:main} and~\ref{lem:d3}.

Hence, the main challenge is proving that $\det(\mathbf{H})>0$. This can be done slightly more easily in the region \eqref{eq:con3} and \eqref{eq:regionlower}. Indeed, in Section~\ref{sec:proof-case13} we prove the following lemma.

\begin{lemma}\label{lem:case13}
$\det(\mathbf{H})>0$ for every $\Delta \geq 3$, $(\alpha,\beta) \in \TD$, $(\gamma,\delta)$ in the interior of~(\ref{eq:con3}) and $(\gamma,\delta)$ in~(\ref{eq:regionlower}).
\end{lemma}

Proving $\det(\mathbf{H})>0$ in the intersection of the regions \eqref{eq:con3} and \eqref{eq:regionlower} is trickier and is essentially the reason we do not obtain our hardness result for $\Delta=4,5$. At this point, it is convenient to split the analysis for each of the lemmas.

\subsection{Proof of Lemma~\ref{lem:main}}
In the setup of Lemma~\ref{lem:main}, we have  $(\alpha,\beta) \in \TD$ and $\alpha,\beta \leq 1/2$. We suppress the details of proving $\det(\mathbf{H})>0$ as a lemma, whose proof we defer to Section \ref{sec:case2}.

\begin{lemma}\label{lem:case2}
$\det(\mathbf{H})>0$ for every $\Delta \geq 3$, $(\alpha,\beta) \in \TD$, $\alpha,\beta \leq 1/2$, $(\gamma,\delta)$ in the interior of~(\ref{eq:con3}) and $(\gamma,\delta)$ in~(\ref{eq:regionupper}).
\end{lemma}
Using Lemmata~\ref{lem:case13} and~\ref{lem:case2}, the proof of Lemma~\ref{lem:main} is immediate.

\begin{proof}[Proof of Lemma~\ref{lem:main}]
Lemma~\ref{lem:case13} and Lemma~\ref{lem:case2} imply that $f$ has a unique maximum at $(\gamma^*,\delta^*)$ for every
$\Delta \geq 3$, $(\alpha,\beta) \in \TD$, $1/2\geq \alpha,\beta$ and $(\gamma,\delta)$
in the interior of~(\ref{eq:con3}). This follows from the fact that $\det(\mathbf{H})>0$ implies that the
Hessian of $f$ is negative definite in the region of interest,
 i.e., where $g_{\alpha,\beta}$ could possibly have stationary points, which in turn implies that $f$ is strictly
 concave in the region and hence has a unique maximum.  By the definition of $f$, it follows that $g_{\alpha,\beta}$ has a unique maximum at $(\gamma^{*},\delta^{*}, \epsilon^{*})$. For a more thorough outline, see the beginning of Section~\ref{sec:conclu}.
\end{proof}

\subsection{Proof of Lemma~\ref{lem:d3}}

In the setup of Lemma~\ref{lem:d3}, we have $\Delta=3$ and $p^+\geq 1/2$. By~(\ref{eq:occupy}), tedious but otherwise simple algebra gives that the solution of
$\beta=\phi(\alpha)$ and $\alpha=\phi(\beta)$ with $\alpha\neq \beta$ satisfies
$$
\alpha^2-2\alpha+\alpha\beta+1-2\beta+\beta^2=0.
$$
It can also be checked that when $1> p^+\geq 1/2$, it holds that $0 < p^- \leq (3-\sqrt{5})/4$. Define $\mathcal{R}'_3$ to be the set of pairs $(\alpha,\beta)$ such
that $\alpha^2-2\alpha+\alpha\beta+1-2\beta+\beta^2 = 0$,
$1/2 \leq \alpha < 1$ and $0 < \beta \leq (3-\sqrt{5})/4$. Our goal is to show that
$\det(\mathbf{H})>0$ for every $(\alpha,\beta) \in \mathcal{R}'_3$,
$(\gamma,\delta)$ in the interior of~(\ref{eq:con3}) and
$(\gamma,\delta)$ in~(\ref{eq:regionupper}).

We can rewrite $\det(\mathbf{H})$ using the formula in~(\ref{eq:det}) as
\begin{eqnarray}
\det(\mathbf{H}) &=& \frac{1}{D}\Big[3 R_1 (U_1+U_2)+U_1U_2\Big]\nonumber\\
&=& \frac{U_1}{D}\Big[3 R_1 (1+U_2/U_1)+U_2\Big],\label{eq:det2}
\end{eqnarray}
where
\begin{equation*}
U_1 = R_8+R_9-2R_6-2R_7, \quad U_2 = R_4+R_5-2R_2-2R_3.
\end{equation*}

The following lemma establishes  technical inequalities on $U_1,U_2,R_1,\dots,R_9$, which are crucial in establishing the positivity of $\det(\mathbf{H})$. Its proof is given in Section \ref{sec:proof-ineqdelta3}.

\begin{lemma}\label{lem:ineqdelta3}
We have $U_1 >0$, $R_4>R_6$, $R_5>R_7$, $R_5>4R_3$, $R_9>4R_7$, and $3R_8/2+R_9 > 9R_2$, for every $(\alpha,\beta) \in \mathcal{R}'_3$,
$(\gamma,\delta)$ in the interior of~(\ref{eq:con3}) and $(\gamma,\delta)$ in~(\ref{eq:regionupper}).
\end{lemma}

With Lemma~\ref{lem:ineqdelta3} at hand, we can now prove that $\det(\mathbf{H})$ is positive.

\begin{lemma}\label{lem:cased3}
$\det(\mathbf{M})>0$ for every $(\alpha,\beta) \in \mathcal{R}'_3$, $(\gamma,\delta)$ in the interior of~(\ref{eq:con3}) and $(\gamma,\delta)$ in~(\ref{eq:regionupper}).
\end{lemma}

\begin{proof}%[Proof of Lemma~\ref{lem:cased3}]
Observe that
\begin{equation*}
\begin{split}
U_1+3U_2 &= (R_8+2R_9/3-6R_2)+(3R_4-2R_6)+(R_9/3-4R_7/3)\\
&\ \ \ +(3R_5/2-2R_7/2)+(3R_5/2-6R_3)\\
&>0,
\end{split}
\end{equation*}
where the last inequality follows by Lemma~\ref{lem:cased3}. Once again by Lemma~\ref{lem:cased3}, we have $U_1>0$ and  hence $U_2/U_1>-1/3$. Thus, (\ref{eq:det2}) gives
$$
\det(\mathbf{H})>\frac{U_1}{D}(2 R_1 + U_2)=\frac{U_1}{D}\big[R_4+R_5+2(R_1-R_2-R_3)\big] > 0,
$$
where the last inequality follows from Lemma~\ref{lem:secondd} and Observation~\ref{cpositive}.
\end{proof}

\begin{proof}[Proof of Lemma~\ref{lem:d3}]
As in the proof of Lemma~\ref{lem:main}, Lemmata~\ref{lem:case13} and~\ref{lem:cased3} yield that $g_{\alpha,\beta}(\gamma,\delta,\epsilon)$ achieves its unique maximum in the region~(\ref{eq:con2}) at the
point $(\gamma^*,\delta^*,\epsilon^*)$, for every $(\alpha,\beta) \in \mathcal{R}'_3$. We next show
that $g_{\alpha,\beta}(\gamma,\delta,\epsilon):=\Phi_2(\alpha,\beta,\gamma,\delta,\epsilon)$
also achieves its unique maximum in the region~(\ref{eq:con2}) at the
point $(\gamma^*,\delta^*,\epsilon^*)$ in a small neighborhood of $\mathcal{R}'_3$.

First note that $\Phi_2$ is continuous. By Lemma~\ref{lem:interior}, we have for sufficiently small $\chi > 0$, the maximum of $g_{\alpha,\beta}$ cannot be obtained on the boundary of the region~(\ref{eq:con2}).

Note that the derivatives of $\Phi_2$ are continuous. It follows that for sufficiently small $\chi > 0$, all stationary points of $g_{\alpha,\beta}$ have to be close to the point $(\gamma^*,\delta^*,\epsilon^*)$. We can choose $\chi$ such that $\det(\mathbf{H})>0$ in the neighborhood of the point $(\gamma^*,\delta^*,\epsilon^*)$, which implies that $g_{\alpha,\beta}$ has a unique stationary point and it is a maximum.
\end{proof}

\section{Remaining Proofs of Technical Lemmas}\label{sec:proofs}

\subsection{Proof of Lemma~\ref{lem:zeros}}\label{sec:lemzeros}
We define $W_3$ to be the numerator of $W_{11}$ minus the denominator of $W_{11}$, and $W_4$ to be the numerator of $W_{21}$ minus the
denominator of $W_{21}$; more precisely
\begin{eqnarray}
W_3 &=& (\alpha-\gamma-\hat{\epsilon})\hat{\epsilon}(1-2\alpha+\gamma) - \hat{\eta}(\alpha-\gamma)^{2},\\
W_4 &=& (\beta-\delta-\hat{\eta})\hat{\eta}(1-2\beta+\delta) - \hat{\epsilon}(\beta-\delta)^{2}.
\end{eqnarray}
Substituting the expression for $\hat{\epsilon}$ and simplifying we obtain
\begin{eqnarray*}
W_3 &:=& ((\alpha+\beta-1+\sqrt{D})/2)\hat{\epsilon}(1-2\alpha+\gamma)-\hat{\eta}(\alpha-\gamma)^{2},\\
W_4 &:=& ((\alpha+\beta-1+\sqrt{D})/2)\hat{\eta}(1-2\beta+\delta)-\hat{\epsilon}(\beta-\delta)^{2}.
\end{eqnarray*}
By expanding $W_3$ and $W_4$, we have
\begin{eqnarray*}
W_3 &=& W_{31}+W_{32}\sqrt{D},\\
W_4 &=& W_{41}+W_{42}\sqrt{D},
\end{eqnarray*}
where
\begin{eqnarray*}
W_{31} &=& ((3/2)\beta-(1/2)\beta^2-\delta-(3/2)\alpha\beta+\delta\alpha)\gamma\\
&&+\beta+(5/2)\beta\alpha^2+(1/2)\alpha^3+(3/2)\alpha-(3/2)\alpha^2\\
&&-(1/2)\beta^2-1/2+\delta\alpha-(7/2)\alpha\beta-\delta\alpha^2+\beta^2\alpha,\\
W_{32} &=& \alpha\beta-(1/2)\beta\gamma-\alpha-(1/2)\beta+(1/2)\alpha^2+1/2,\\
W_{41} &=& ((3/2)\alpha-(1/2)\alpha^2-\gamma-(3/2)\alpha\beta+\gamma\beta)\delta\\
&&+\alpha+(5/2)\alpha\beta^2+(1/2)\beta^3+(3/2)\beta-(3/2)\beta^2\\
&&-(1/2)\alpha^2-1/2+\gamma\beta-(7/2)\alpha\beta-\gamma\beta^2+\alpha^2\beta,\\
W_{42} &=& \alpha\beta-(1/2)\alpha\delta-\beta-(1/2)\alpha+(1/2)\beta^2+1/2.
\end{eqnarray*}
Note that for every $(\alpha,\beta) \in \mathcal{R}$, $W_{32} > 0$ when $0<\gamma<\alpha$, and $W_{42} > 0$ when $0<\delta<\beta$. To
see $W_{32}>0$ note that $\alpha\beta-(1/2)\beta\gamma-\alpha-(1/2)\beta+(1/2)\alpha^2+1/2\geq
\alpha\beta-(1/2)\beta\alpha-\alpha-(1/2)\beta+(1/2)\alpha^2+1/2=(1/2)(1-\alpha)(1-\alpha-\beta)>0$, since $\alpha+\beta<1$;
inequality $W_{42}>0$ is the same (after renaming the variables).

Note that $W_{31}$ is a linear function in $\gamma$ and we have
$$
\frac{d W_{31}}{d \gamma} = (3/2)\beta-(1/2)\beta^2-\delta-(3/2)\alpha\beta+\delta\alpha,
$$
which is positive for all $(\alpha,\beta) \in \mathcal{R}$ and $0<\delta<\beta$. To see this note
\begin{equation*}
\begin{split}
(3/2)\beta-(1/2)\beta^2-\delta-(3/2)\alpha\beta+\delta\alpha=(1-\alpha)((3/2)\beta-\delta) - (1/2)\beta^2\\
\geq (1-\alpha)(1/2)\beta - (1/2)\beta ^2 = (1/2) \beta (1-\alpha-\beta)>0.
\end{split}
\end{equation*}
Moreover, we have $W_{31}$ is negative when $\gamma=\alpha$, $(\alpha,\beta) \in \mathcal{R}$ and $0<\delta<\beta$ (after substituting
$\gamma=\alpha$ into $W_{31}$ we obtain $-(1-\alpha)(1-\alpha-\beta)^2/2<0$). Hence, $W_{31}<0$ for all $(\alpha,\beta) \in
\mathcal{R}$, $0<\gamma<\alpha$ and $0<\delta<\beta$.

By the same proof, we can show that $W_{41}<0$ for all $(\alpha,\beta) \in \mathcal{R}$, $0<\gamma<\alpha$ and $0<\delta<\beta$ (note
that $W_{31}$ and $W_{41}$ are the same after renaming the variables).

Let
\begin{eqnarray*}
W_5 &=& (W_{31}/W_{32})^2-D,\\
W_6 &=& (W_{41}/W_{42})^2-D.
\end{eqnarray*}
Note that the signs of $W_3$ and $W_5$ are opposite, and the signs of $W_4$ and $W_6$ are opposite.

After substituting $W_{31},W_{32},W_{41},W_{42}$ and simplifications, we obtain
\begin{eqnarray*}
W_5 &=& -\frac{4(\beta-\delta)(\alpha-\gamma)^2((1-\alpha)^2\delta+\beta^2(2\alpha-1-\gamma))}{(2\alpha\beta-\beta\gamma-2\alpha-\beta+\alpha^2+1)^2},\\
W_6 &=&
-\frac{4(\beta-\delta)^2(\alpha-\gamma)((1-\beta)^2\gamma+\alpha^2(2\beta-1-\delta))}{(\beta^2-\delta\alpha+1-2\beta-\alpha+2\alpha\beta)^2}.
\end{eqnarray*}

Hence, $$W_{11} \geq 1 \iff W_5 \leq 0 \iff  (1-\alpha)^2\delta+\beta^2(2\alpha-1-\gamma) \geq 0,$$ and
$$W_{21} \geq 1 \iff W_6 \leq 0 \iff (1-\beta)^2\gamma+\alpha^2(2\beta-1-\delta) \geq 0.$$

FInally, we analyze the conditions for $W_{12}\geq 1$ and $W_{22}\geq 1$. It is straightforward to see that
$$
W_{12} \geq 1 \iff (\alpha-\gamma)^2 \geq (1-2\alpha+\gamma)\gamma \iff \gamma \leq \alpha^2,
$$
and
$$
W_{22} \geq 1 \iff (\beta-\delta)^2 \geq (1-2\beta+\delta)\delta \iff \delta \leq \beta^2.
$$

\subsection{Proof of Lemma \ref{lem:secondd}}
\label{sec:proof-secondd}

We have
$$
-R_1+R_2+R_3 = (\beta-\delta)\left(-\frac{1}{\alpha-\gamma-\hat{\epsilon}}-\frac{1}{\hat{\eta}}+\frac{1}{\hat{\epsilon}}\right)
-\frac{\sqrt{D}}{\hat{\epsilon}}+\frac{\sqrt{D}}{1-2\alpha+\gamma}+\frac{2\sqrt{D}}{\alpha-\gamma}.
$$
Multiplying by the denominators we let
$$
P_1 =(-R_1+R_2+R_3) \hat{\epsilon}\hat{\eta}(1-2\alpha+\gamma)(\alpha-\gamma)(\alpha-\gamma-\hat{\epsilon}) =
P_{11}+P_{12}\sqrt{D},
$$
where
\begin{eqnarray*}
P_{11} &=& 1-(1/2)\delta\alpha+17\alpha\beta+(9/2)\delta\gamma-5\beta\gamma+21\alpha\gamma\beta\delta-(37/2)\alpha^2\gamma\beta\delta-5\alpha\gamma\beta\delta^2+(5/2)\alpha\gamma\beta^2\delta\\
&&+8\gamma^2\alpha\delta\beta-5\alpha-3\beta-\delta+10\alpha^2+3\beta^2+3\delta\beta-10\alpha^3+5\alpha^4-\beta^3+15\alpha\gamma\beta-(13/2)\alpha\beta\delta\\
&&-(27/2)\alpha\gamma\delta-(5/2)\beta\gamma\delta-4\alpha^2\beta\delta+(27/2)\alpha^2\gamma\delta-15\alpha^2\gamma\beta-24\alpha\gamma\beta^2-8\gamma^2\beta\delta\\
&&+(17/2)\alpha\beta^2\delta-(5/2)\beta^2\gamma\delta-3\alpha\beta\delta^2+3\beta\gamma\delta^2+(15/2)\alpha^3\beta\delta+4\alpha^2\beta\delta^2-(11/2)\alpha^2\beta^2\delta\\
&&-(9/2)\alpha^3\gamma\delta+3\alpha^2\gamma\delta^2+5\alpha^3\gamma\beta+17\alpha^2\gamma\beta^2+4\alpha\gamma\beta^3-3\gamma^2\alpha\delta^2-5\gamma^2\alpha\beta^2+\gamma^2\beta\delta^2-3\beta^2\delta\\
&&+3\delta^2\alpha-3\delta^2\gamma-3\alpha^2\delta^2+(7/2)\alpha^4\delta-8\alpha^4\beta-15\alpha^3\beta^2-4\alpha^2\beta^3+3\gamma^2\delta^2-\gamma^2\beta^3-33\alpha^2\beta\\
&&+(15/2)\alpha^2\delta-16\alpha\beta^2+7\beta^2\gamma-(19/2)\alpha^3\delta+27\alpha^3\beta+28\alpha^2\beta^2+5\gamma^2\beta^2+4\alpha\beta^3-2\beta^3\gamma\\
&&+\delta\beta^3-(3/2)\alpha\beta^3\delta+(1/2)\beta^3\gamma\delta-\alpha^5,
\end{eqnarray*}
and
\begin{eqnarray*}
P_{12} &=& -1-(1/2)\delta\alpha-9\alpha\beta-(5/2)\delta\gamma+3\beta\gamma-2\alpha\gamma\beta\delta+4\alpha+2\beta+\delta-6\alpha^2-\beta^2-2\delta\beta\\
&&+4\alpha^3-\alpha^4-6\alpha\gamma\beta+4\alpha\beta\delta+5\alpha\gamma\delta-\alpha^2\beta\delta-(5/2)\alpha^2\gamma\delta+3\alpha^2\gamma\beta+4\alpha\gamma\beta^2+\gamma^2\beta\delta\\
&&-(3/2)\alpha\beta^2\delta+(1/2)\beta^2\gamma\delta-\alpha\gamma\delta^2+\beta^2\delta-\delta^2\alpha+\delta^2\gamma+\alpha^2\delta^2+12\alpha^2\beta-2\alpha^2\delta+4\alpha\beta^2\\
&&-2\beta^2\gamma+(3/2)\alpha^3\delta-5\alpha^3\beta-4\alpha^2\beta^2-\gamma^2\beta^2.
\end{eqnarray*}

The following claim is proved using the {\tt Resolve} function of the Mathematica system in Appendix~\ref{app:secondd}.

\begin{claim}\label{cccf}
$P_{11}>0$ and $P_{12}<0$ for all $(\alpha,\beta) \in \mathcal{R}$ and $(\gamma,\delta)$ in the interior of~(\ref{eq:con3}).
\end{claim}

From Claim~\ref{cccf} we have that
$P_{11}-P_{12}\sqrt{D}>0$ and hence showing $P_{11}+P_{12}\sqrt{D}<0$ is equivalent to showing
$P_{11}^2 - P_{12}^2 D<0$ which in turn is equivalent to showing that $P_2 := (P_{11}/P_{12})^2-D$ is negative. We have
$$
P_2 = (P_{11}/P_{12})^2-D = -\frac{(\beta-\delta)^2(\alpha-\gamma)^3P_{21}}{P_{12}^2},
$$
where
\begin{eqnarray*}
P_{21} &=& -8\delta\alpha+20\alpha\gamma\beta\delta-10\alpha^2\gamma\beta\delta+2\alpha\gamma\beta\delta^2+12\alpha\gamma\beta^2\delta+2\delta+2\delta^2-8\delta\beta+4\beta^3+34\alpha\beta\delta\\
&&-10\beta\gamma\delta-44\alpha^2\beta\delta+8\alpha\beta^2\delta-4\beta^2\gamma\delta-14\alpha\beta\delta^2-2\beta\gamma\delta^2+18\alpha^3\beta\delta+6\alpha^2\beta\delta^2-10\alpha^2\beta^2\delta\\
&&+9\alpha^2\gamma\delta^2-16\alpha\gamma\beta^3-18\alpha\gamma\delta^2-2\beta^2\delta-15\delta^2\alpha+9\delta^2\gamma+24\alpha^2\delta^2+2\alpha^4\delta+16\alpha^2\beta^3\\
&&+4\gamma^2\beta^3+12\alpha^2\delta-8\alpha^3\delta-16\alpha\beta^3+8\beta^3\gamma-4\delta\beta^3+6\alpha\beta^3\delta-2\beta^3\gamma\delta+8\delta^2\beta+2\beta^2\delta^2\\
&&+8\delta^3\alpha-4\delta^3\alpha^2-11\delta^2\alpha^3-3\delta^2\alpha\beta^2+\delta^2\gamma\beta^2-4\delta\gamma^2\beta^2-4\delta^3.
\end{eqnarray*}

The following claim is proved using the {\tt Resolve} function of the Mathematica system in Appendix~\ref{sec:ccch}.

\begin{claim}\label{ccch}
$P_{21}>0$ for all $(\alpha,\beta) \in \mathcal{R}$ and $(\gamma,\delta)$ in the interior of~(\ref{eq:con3}).
\end{claim}

From Claim~\ref{ccch} we have that $P_2<0$ which implies $P_{11}+P_{12}\sqrt{D}<0$ and this completes the
proof of Lemma \ref{lem:secondd}.

\subsection{Proof of Lemma \ref{lem:case13}}
\label{sec:proof-case13}

We will use the following technical lemma in the proof of Lemma \ref{lem:case13}.

\begin{lemma}\label{lem:detcoeff}
For every $(\alpha,\beta) \in \mathcal{R}$, and $(\gamma,\delta)$ in the interior of~(\ref{eq:con3}), we have
\begin{equation}\label{ett1}
(-R_1+R_2+R_3)(-R_1+R_6+R_7)-R_1^2<0,
\end{equation}
\begin{equation}\label{ett2}
(-R_1+R_2+R_3)(-R_1-R_8-R_9)-R_1^2 \geq 0,
\end{equation}
and
\begin{equation}\label{ett3}
(-R_1+R_6+R_7)(-R_1-R_4-R_5)-R_1^2 \geq 0.
\end{equation}
\end{lemma}

\begin{proof}[Proof of Lemma~\ref{lem:detcoeff}]
We first prove inequality \eqref{ett1}. We have that
$$
(-R_1+R_2+R_3)(-R_1+R_6+R_7)=(R_1-R_2-R_3)(R_1-R_6-R_7)<R^2_1,
$$
where the last inequality uses Lemma~\ref{lem:secondd} and the positivity of the $R_i$ (Observation~\ref{cpositive}). This establishes \eqref{ett1}.

We next prove \eqref{ett2}, noting that inequality~\eqref{ett3} follows by an analogous argument. The left-hand side of~\eqref{ett2} multiplied by the denominators and simplified is
\begin{eqnarray*}
P_5 &=& ((-R_1+R_2+R_3)(-R_1-R_8-R_9)-R_1^2)\delta(\alpha-\gamma)\hat{\epsilon}\hat{\eta}(1-2\alpha+\gamma)(\alpha-\gamma-\hat{\epsilon})/\sqrt{D}\\
&=& P_{51}+P_{52}\sqrt{D},
\end{eqnarray*}
where
\begin{eqnarray*}
P_{51} &=& -1-4\beta\gamma^2\delta\alpha+3\delta\alpha-17\alpha\beta-3\delta\gamma+5\beta\gamma-14\alpha\beta\gamma\delta+11\alpha^2\beta\gamma\delta+5\alpha+3\beta-10\alpha^2-3\beta^2\\
&&-15\alpha\beta\gamma+9\delta\alpha\gamma-3\alpha\beta\delta+10\alpha^2\beta\delta+15\alpha^2\beta\gamma+24\alpha\beta^2\gamma+3\beta\gamma\delta+4\beta\gamma^2\delta-9\delta\alpha^2\gamma\\
&&-7\alpha^3\beta\delta-5\alpha^3\beta\gamma-17\alpha^2\beta^2\gamma-4\alpha\beta^3\gamma+5\beta^2\gamma^2\alpha+3\delta\alpha^3\gamma+33\alpha^2\beta-9\delta\alpha^2+16\alpha\beta^2\\
&&-27\alpha^3\beta-28\alpha^2\beta^2-7\beta^2\gamma-5\beta^2\gamma^2+9\delta\alpha^3-4\alpha\beta^3+8\alpha^4\beta+15\alpha^3\beta^2+4\alpha^2\beta^3+2\beta^3\gamma\\
&&+\beta^3\gamma^2-3\delta\alpha^4+10\alpha^3-5\alpha^4+\beta^3+\alpha^5,
\end{eqnarray*}
and
\begin{eqnarray*}
P_{52} &=& 1-2\delta\alpha\gamma+6\alpha\beta\gamma-3\alpha^2\beta\gamma-4\alpha\beta^2\gamma+\delta\alpha^2\gamma-4\alpha+6\alpha^2-2\beta-4\alpha^3+\alpha^4+\beta^2+9\alpha\beta\\
&&-3\beta\gamma-\delta\alpha+\delta\gamma+2\delta\alpha^2-12\alpha^2\beta-4\alpha\beta^2+5\alpha^3\beta+4\alpha^2\beta^2+2\beta^2\gamma+\beta^2\gamma^2-\delta\alpha^3.
\end{eqnarray*}

The following claim is proved using the {\tt Resolve} function of the Mathematica system in Appendix~\ref{sec:detcoeff}.

\begin{claim}\label{ccccr}
$P_{51}<0$, $P_{52} > 0$ for all all $(\alpha,\beta) \in \mathcal{R}$ and $(\gamma,\delta)$ in the interior of~(\ref{eq:con3}).
\end{claim}

From Claim~\ref{ccccr} we have that $P_{52}\sqrt{D}-P_{51}>0$ and hence the sign of $P_5=P_{51}+P_{52}\sqrt{D}$ is the same as
the sign of $D P_{52}^2 - P_{51}^2$ which is the same as the sign of
\begin{eqnarray*}
D - (P_{51}/P_{52})^2 =
\frac{4(\beta-\delta)(\alpha-\gamma)^3\left((1-\alpha)^2\delta+\beta^2(2\alpha-1-\gamma)\right)^2}{P_{52}^2} > 0.
\end{eqnarray*}
Hence $P_5>0$ which implies~\eqref{ett2}.
\end{proof}

\begin{proof}[Proof of Lemma~\ref{lem:case13}]
From $\alpha+\beta+\Delta(\Delta-2)\alpha\beta\leq 1$ we have
\begin{equation}\label{eeeeeek}
(\Delta-1)^2\leq (1-\alpha)(1-\beta)/(\alpha\beta)=:R_{11}.
\end{equation}
Lemma~\ref{lem:detcoeff} implies that the coefficient of $(\Delta-1)$ in (\ref{eq:det}) is positive and the coefficient of
$(\Delta-1)^2$ in (\ref{eq:det}) is negative. Hence, also using~\eqref{eeeeeek}, we have
\begin{eqnarray*}
\det(M) &\geq& \frac{1}{D}\Big\{R_{11}\left[(-R_1+R_2+R_3)(-R_1+R_6+R_7)-R_1^2\right]+\Big.\\
&&\hskip 3.5cm +\Big.\big[(-R_1-R_8-R_9)(-R_1-R_4-R_5)-R_1^2\big]\Big\}\\
&>& \frac{1}{D}\Big(-R_1R_2R_{11}-R_1R_6R_{11}+R_1R_4+R_1R_8+\Big.\\
&&\hskip 0.8cm+R_5R_9+R_3R_7R_{11}+R_1R_5+R_1R_9-R_1R_3R_{11}-R_1R_7R_{11}+\\
&&\hskip 0.8cm+\Big.R_2R_7R_{11}+R_3R_6R_{11}\Big).
\end{eqnarray*}

We will show
\begin{equation}\label{eq:c131}
-R_1R_2R_{11}-R_1R_6R_{11}+R_1R_4+R_1R_8>0,
\end{equation}
and
\begin{equation}\label{eq:c132}
R_5R_9+R_3R_7R_{11}+R_1R_5+R_1R_9-R_1R_3R_{11}-R_1R_7R_{11}+R_2R_7R_{11}+R_3R_6R_{11}>0.
\end{equation}

To prove~(\ref{eq:c131}), note that
\begin{equation*}
\begin{split}
-R_1R_2R_{11}-R_1R_6R_{11}+R_1R_4+R_1R_8 = \\
\frac{\sqrt{D}(1-\alpha-\beta)}{1-\alpha-\beta-\hat{\epsilon}-\hat{\eta}}\left(\frac{1}{\gamma}+\frac{1}{\delta}-\frac{(1-\alpha)(1-\beta)}{\alpha\beta(1-2\alpha+\gamma)}-\frac{(1-\alpha)(1-\beta)}{\alpha\beta(1-2\beta+\delta)}\right).
\end{split}
\end{equation*}
The positivity of (\ref{eq:c131}) follows from the following claim.

\begin{claim}\label{clm:s1}
$$
\frac{1}{\gamma}+\frac{1}{\delta}-\frac{(1-\alpha)(1-\beta)}{\alpha\beta(1-2\alpha+\gamma)}-\frac{(1-\alpha)(1-\beta)}{\alpha\beta(1-2\beta+\delta)}
> 0,
$$
for every $(\alpha,\beta) \in \mathcal{R}_3\supset \TD$, $(\gamma,\delta)$ in the interior of~(\ref{eq:con3}), and $(\gamma,\delta)$
in~(\ref{eq:regionlower}).
\end{claim}

\begin{proof}[Proof of Claim~\ref{clm:s1}]
\begin{eqnarray*}
&&\left(\frac{1}{\gamma}+\frac{1}{\delta}-\frac{(1-\alpha)(1-\beta)}{\alpha\beta(1-2\alpha+\gamma)}-\frac{(1-\alpha)(1-\beta)}{\alpha\beta(1-2\beta+\delta)}\right)\alpha\beta\gamma\delta(1-2\alpha+\gamma)(1-2\beta+\delta)\\
&=& -2\gamma\delta+\delta\alpha\beta+\gamma\alpha\beta-4\delta\alpha\beta\gamma-2\delta\alpha\beta^2+\delta^2\alpha\beta-2\delta\alpha^2\beta+4\delta\alpha^2\beta^2-2\delta^2\alpha^2\beta-2\gamma\alpha\beta^2\\
&&-2\gamma\alpha^2\beta+4\gamma\alpha^2\beta^2+\gamma^2\alpha\beta-2\gamma^2\alpha\beta^2+4\gamma\delta\beta-2\gamma\delta\beta^2+\gamma\delta^2\beta+4\gamma\delta\alpha+\gamma\delta^2\alpha+\gamma^2\delta\beta\\
&&-2\gamma\delta\alpha^2+\gamma^2\delta\alpha-\gamma\delta^2-\gamma^2\delta > 0,
\end{eqnarray*}
where the last inequality is proved using {\tt Resolve} function of Mathematica system, see Appendix~\ref{sec:s1}.
\end{proof}

To show~(\ref{eq:c132}), we first note that
\begin{eqnarray*}
&&R_5R_9+R_3R_7R_{11}+R_1R_5+R_1R_9-R_1R_3R_{11}-R_1R_7R_{11}+R_2R_7R_{11}+R_3R_6R_{11}\\
&=&\sqrt{D}\left(\frac{4(1-\beta-\gamma-\hat{\epsilon})}{(\alpha-\gamma)}+\frac{(1-\alpha)(1-\beta)}{(1-2\alpha+\gamma)\alpha\beta}\frac{2(\alpha-\gamma-\hat{\epsilon})}{\beta-\delta}+\frac{(1-\alpha)(1-\beta)}{(1-2\beta+\delta)\alpha\beta}\frac{2(\alpha-\gamma-\hat{\epsilon})}{\alpha-\gamma}\right.\\
&&\left.-\frac{2(\alpha-\gamma-\hat{\epsilon})^2}{(\alpha-\gamma)(\beta-\delta)}\left(\frac{1}{\hat{\eta}}+\frac{1}{\hat{\epsilon}}\right)\left(\frac{(1-\alpha)(1-\beta)}{\alpha\beta}-\frac{1-\beta-\gamma-\hat{\epsilon}}{\alpha-\gamma-\hat{\epsilon}}\right)\right)\\
&>&\frac{2\sqrt{D}(\alpha-\gamma-\hat{\epsilon})}{\beta-\delta}\left(\frac{(1-\alpha)(1-\beta)}{(1-2\alpha+\gamma)\alpha\beta}-\frac{\alpha-\gamma-\hat{\epsilon}}{(\alpha-\gamma)\hat{\eta}}\left(\frac{(1-\alpha)(1-\beta)}{\alpha\beta}-\frac{1-\beta-\gamma-\hat{\epsilon}}{\alpha-\gamma-\hat{\epsilon}}\right)\right)\\
&&+\frac{2\sqrt{D}(\alpha-\gamma-\hat{\epsilon})}{\alpha-\gamma}\left(\frac{(1-\alpha)(1-\beta)}{(1-2\beta+\delta)\alpha\beta}-\frac{\alpha-\gamma-\hat{\epsilon}}{(\beta-\delta)\hat{\epsilon}}\left(\frac{(1-\alpha)(1-\beta)}{\alpha\beta}-\frac{1-\beta-\gamma-\hat{\epsilon}}{\alpha-\gamma-\hat{\epsilon}}\right)\right).
\end{eqnarray*}
We will show
$$
\frac{(1-\alpha)(1-\beta)}{(1-2\alpha+\gamma)\alpha\beta}-\frac{\alpha-\gamma-\hat{\epsilon}}{(\alpha-\gamma)\hat{\eta}}\left(\frac{(1-\alpha)(1-\beta)}{\alpha\beta}-\frac{1-\beta-\gamma-\hat{\epsilon}}{\alpha-\gamma-\hat{\epsilon}}\right)>0,
$$
for every $(\alpha,\beta) \in \mathcal{R}_3\supset \TD$, $(\gamma,\delta)$ in the interior of~(\ref{eq:con3}), and $(\gamma,\delta)$
in~(\ref{eq:regionlower}). Then by symmetry, we have
$$
\frac{(1-\alpha)(1-\beta)}{(1-2\beta+\delta)\alpha\beta}-\frac{\alpha-\gamma-\hat{\epsilon}}{(\beta-\delta)\hat{\epsilon}}\left(\frac{(1-\alpha)(1-\beta)}{\alpha\beta}-\frac{1-\beta-\gamma-\hat{\epsilon}}{\alpha-\gamma-\hat{\epsilon}}\right)>0.
$$

Let
\begin{eqnarray*}
P_7
&=&\left(\frac{(1-\alpha)(1-\beta)}{(1-2\alpha+\gamma)\alpha\beta}-\frac{\alpha-\gamma-\hat{\epsilon}}{(\alpha-\gamma)\hat{\eta}}\left(\frac{(1-\alpha)(1-\beta)}{\alpha\beta}-\frac{1-\beta-\gamma-\hat{\epsilon}}{\alpha-\gamma-\hat{\epsilon}}\right)\right)\\
&&\alpha\beta\hat{\eta}(\alpha-\gamma)(1-2\alpha+\gamma)\\
&=& P_{71}+P_{72}\sqrt{D},
\end{eqnarray*}
where
\begin{eqnarray*}
P_{71} &=& (3/2)\alpha^2+4\alpha\beta+(1/2)\beta^2-\alpha\delta-\gamma\beta+\gamma\delta+1/2-(3/2)\alpha-\beta+\delta\alpha\beta+(3/2)\gamma\alpha\beta\\
&&+\delta\alpha\beta\gamma-\delta\alpha^2\beta-(3/2)\gamma\alpha\beta^2-(1/2)\gamma\alpha^2\beta-\gamma\delta\beta-\gamma\delta\alpha-(1/2)\alpha^3-(9/2)\alpha^2\beta\\
&&-(5/2)\alpha\beta^2+\gamma\beta^2+\alpha^2\delta+(3/2)\alpha^3\beta+(5/2)\alpha^2\beta^2,
\end{eqnarray*}
and
\begin{eqnarray*}
P_{72} = (1/2)\gamma\alpha\beta-1/2+(1/2)\beta+\alpha-(1/2)\alpha^2-(1/2)\alpha\beta-(1/2)\alpha^2\beta.
\end{eqnarray*}

The following claim is proved using the {\tt Resolve} function of the Mathematica system in Appendix~\ref{sec:ckkk}.

\begin{claim}\label{ckkk}
$P_{71}>0$ and $P_{72}<0$ for every $(\alpha,\beta) \in \mathcal{R}_3\supset \TD$, $(\gamma,\delta)$ in the interior of~(\ref{eq:con3}), and
$(\gamma,\delta)$ in~(\ref{eq:regionlower}).
\end{claim}

From Claim~\ref{ckkk} we have that the following expression has the same sign as $P_7$:
\begin{eqnarray*}
(P_{71}/P_{72})^2-D = \frac{P_{71}^2-P_{72}^2D}{P_{72}^2},
\end{eqnarray*}
where
\begin{eqnarray*}
P_{71}^2-P_{72}^2D > 0.
\end{eqnarray*}
The last inequality is true for every every $(\alpha,\beta) \in \mathcal{R}_3\supset \TD$, $(\gamma,\delta)$ in the interior of~(\ref{eq:con3}), and
$(\gamma,\delta)$ in~(\ref{eq:regionlower}); it is proved  using the {\tt Resolve} function of the Mathematica system, see Appendix~\ref{app:case13}.
\end{proof}

\subsection{Proof of Lemma \ref{lem:case2}}
\label{sec:case2}

We first prove the following upper bound on $(\Delta-1)^2$.
\begin{proposition}\label{pro:bigdelta2}
For every $\Delta \geq 3$, $(\alpha,\beta) \in \TD$, $(\gamma,\delta)$ in the interior of~(\ref{eq:con3}), and $(\gamma,\delta)$
in~(\ref{eq:regionupper}), we have
$$
(\Delta-1)^2 \leq \frac{(1-\alpha)(1-\beta)}{\alpha\beta} < \frac{1-\beta-\gamma-\hat{\epsilon}}{\alpha-\gamma-\hat{\epsilon}}
$$
\end{proposition}

\begin{proof}[Proof of Proposition~\ref{pro:bigdelta2}]
The first inequality is straightforward from $\alpha+\beta+\Delta(\Delta-2)\alpha\beta \leq 1$. We next prove that
$(1-\alpha)(1-\beta)(\alpha-\gamma-\hat{\epsilon})<\alpha\beta(1-\beta-\gamma-\hat{\epsilon})$. We have
\begin{eqnarray*}
&&\alpha\beta(1-\beta-\gamma-\hat{\epsilon})-(1-\alpha)(1-\beta)(\alpha-\gamma-\hat{\epsilon})\\
&=& \frac{1}{2}(1-\alpha-\beta)(1-\alpha-\beta+2\alpha\beta-\sqrt{D})>0,\\
\end{eqnarray*}
where the last inequality follows from
$$D=(1-\alpha-\beta)^2+4(\alpha-\gamma)(\beta-\delta)<(1-\alpha-\beta+2\alpha\beta)^2.$$
The last inequality is true when $\gamma>\alpha^2$ and $\delta>\beta^2$ (to see this note
$(1-\alpha-\beta+2\alpha\beta)^2-(1-\alpha-\beta)^2=4(\alpha-\alpha^2)(\beta-\beta^2)>4(\alpha-\gamma)(\beta-\delta)$).
\end{proof}

Let $R_{10}=(1-\beta-\gamma-\hat{\epsilon})/(\alpha-\gamma-\hat{\epsilon})$. Note that $R_7R_{10}=R_9$, $R_3R_{10}=R_5$.

Lemma~\ref{lem:detcoeff} implies that the coefficient of $(\Delta-1)$ in (\ref{eq:det}) is positive and the coefficient of
$(\Delta-1)^2$ in (\ref{eq:det}) is negative. Hence, also using~Proposition~\ref{pro:bigdelta2}, we have

\begin{eqnarray*}
\det(M) &\geq& \frac{1}{D}\left(R_{10}\left((-R_1+R_2+R_3)(-R_1+R_6+R_7)-R_1^2\right)\right.\\
&&+\left.\left((-R_1-R_8-R_9)(-R_1-R_4-R_5)-R_1^2\right)\right),\\
&>& \frac{1}{D}\left(-R_1R_2R_{10}-R_1R_6R_{10}+R_1R_4+R_1R_8+R_5R_9+R_3R_7R_{10}+R_2R_9+R_5R_6\right).
\end{eqnarray*}
We will prove that for every $(\alpha,\beta) \in \mathcal{R}_3\supset \TD$, $\alpha,\beta \leq 1/2$, $(\gamma,\delta)$ in the interior
of~(\ref{eq:con3}), and $(\gamma,\delta)$ in~(\ref{eq:regionupper}),
\begin{equation}\label{eq:c21}
R_1R_4-R_1R_2R_{10}+\frac{1}{2}(R_5R_9+R_3R_7R_{10})+R_2R_9>0,
\end{equation}
and then by symmetry, we also have
$$
R_1R_8-R_1R_6R_{10}+\frac{1}{2}(R_5R_9+R_3R_7R_{10})+R_5R_6>0.
$$
The left-hand side of~(\ref{eq:c21}) can be written as the sum of the following two terms:
\begin{equation}\label{eq:c22}
\frac{(1-\alpha-\beta)\sqrt{D}}{1-\alpha-\beta-\hat{\epsilon}-\hat{\eta}}\left(\frac{1}{\gamma}-\frac{1}{1-2\alpha+\gamma}\right),
\end{equation}
and
\begin{equation}\label{eq:c23}
\frac{\sqrt{D}}{(1-2\alpha+\gamma)\hat{\eta}}\left(\frac{2(1-\alpha)(1-\alpha-\beta-\hat{\epsilon})}{\alpha-\gamma}-\frac{(1-\alpha-\beta)^2}{\hat{\epsilon}}\right).
\end{equation}

We have~(\ref{eq:c22}) is non-negative, since $1-2\alpha \geq 0$.

For~(\ref{eq:c23}), we only need to prove
$$
\frac{1-\alpha-\beta-\hat{\epsilon}}{\alpha-\gamma}>\frac{(1-\alpha-\beta)^2}{\hat{\epsilon}},
$$
since $1-2\alpha \geq 0$. By~(\ref{eq:usefulineq1}), it is equivalent to prove
\begin{equation}\label{eq:c24}
1-\alpha-\beta-\hat{\epsilon}-\hat{\eta}-(1-\alpha-\beta)^2>0.
\end{equation}
To prove~(\ref{eq:c24}), we have
$$
1-\alpha-\beta-\hat{\epsilon}-\hat{\eta}-(1-\alpha-\beta)^2= -(1-\alpha-\beta)^2-\alpha-\beta+\gamma+\delta+\sqrt{D},
$$
and
\begin{eqnarray*}
D - (-(1-\alpha-\beta)^2-\alpha-\beta+\gamma+\delta)^2
&=& 2\gamma+2\delta-2\alpha^2-2\beta^2-6\alpha\delta-6\gamma\beta+2\gamma\delta-2\beta\delta-\delta^2\\
&&+6\alpha\beta^2+6\alpha^2\beta-6\alpha^2\beta^2-2\alpha\gamma-4\alpha^3\beta+2\alpha^2\gamma+2\alpha^2\delta\\
&&-4\alpha\beta^3+2\beta^2\gamma+2\beta^2\delta-\gamma^2+4\alpha\beta\delta+4\alpha\beta\gamma+2\alpha^3\\
&&+2\beta^3-\alpha^4-\beta^4.
\end{eqnarray*}

The following claim is proved using the {\tt Resolve} function of the Mathematica system in Appendix~\ref{sec:qwer}.

\begin{claim}\label{clm:qwer}
$$
D - (-(1-\alpha-\beta)^2-\alpha-\beta+\gamma+\delta)^2>0
$$
for every $(\alpha,\beta) \in \mathcal{R}$, $\alpha \leq 1/2$, $\beta \leq 1/2$, $(\gamma,\delta)$ in the interior of~(\ref{eq:con3}),
and $(\gamma,\delta)$ in~(\ref{eq:regionupper}).
\end{claim}

Since $-(1-\alpha-\beta)^2-\alpha-\beta+\gamma+\delta<0$, we can conclude that~(\ref{eq:c24}) is true. Hence, we complete the
proof of Lemma \ref{lem:case2}.

\subsection{Proof of Lemma \ref{lem:ineqdelta3}}
\label{sec:proof-ineqdelta3}

We have
$$
R_4-R_6 = \sqrt{D} \cdot \frac{1-2\beta+\delta-\gamma}{\gamma(1-2\beta+\delta)} > 0,
$$
where the inequality follows from~(\ref{eq:con3}).

We have
$$
R_5-R_7 =
\frac{(1-\alpha-\beta)(\alpha+\beta-\gamma-\delta)+(\beta-\delta-\alpha+\gamma)\sqrt{D}}{(\alpha-\gamma)(\beta-\delta)}.
$$
If $\beta-\delta-\alpha+\gamma \geq 0$, then $R_5>R_7$. We may assume that $\beta-\delta-\alpha+\gamma < 0$. To show $R_5>R_7$,
it is equivalent to prove that
$$
(1-\alpha-\beta)(\alpha+\beta-\gamma-\delta)>(\alpha-\gamma-\beta+\delta)\sqrt{D}.
$$
We have
$$
\frac{(1-\alpha-\beta)^2(\alpha+\beta-\gamma-\delta)^2}{(\alpha-\gamma-\beta+\delta)^2}-D =
\frac{4(\alpha-\gamma)(\beta-\delta)(1-2\beta-\gamma+\delta)(1-2\alpha-\delta+\gamma)}{(\alpha-\gamma-\beta+\delta)^2}>0,
$$
where the inequality follows from~(\ref{eq:con3}).

To show that $R_5>4R_3$ and $R_9>4R_7$, it is sufficient to prove that
$1-\beta-\gamma-\hat{\epsilon}>4(\alpha-\gamma-\hat{\epsilon})$, which follows from Proposition~\ref{pro:bigdelta2}.

We next prove $U_1>0$. We first prove that $R_8>2R_6$. We have
$$
R_8-2R_6 = \sqrt{D} \cdot \frac{1-2\beta-\delta}{\delta(1-2\beta+\delta)} \geq \sqrt{D} \cdot
\frac{1-3\beta}{\delta(1-2\beta+\delta)} > 0,
$$
where the first inequality follows from~(\ref{eq:con3}) and the last inequality follows from the fact that $\beta \leq
(3-\sqrt{5})/4$. Hence $U_1=R_8-2R_6+R_9-2R_7>0$.

We next show:
\[
\frac{3R_8}{2}+R_9 > 9R_2.
\]
by the fact that $\hat{\epsilon}:= \frac{1}{2}(1+\alpha-\beta-2\gamma - \sqrt{D})$ and $(1-\alpha-\beta)^2=\alpha\beta$, the goal
is then to show,
\[
\frac{3\sqrt{D}}{2\delta}+\frac{\sqrt{\alpha\beta}+\sqrt{D}}{\beta-\delta} > \frac{9\sqrt{D}}{1-2\alpha+\gamma},
\]
i.e.,
\[
\frac{\sqrt{\alpha\beta}}{\sqrt{\alpha\beta+4(\alpha-\gamma)(\beta-\delta)}} > \frac{9(\beta-\delta)}{1-2\alpha+\gamma} -
\frac{3\beta}{2\delta}+1/2.
\]
Since the above inequality is monotone with respect to $\gamma$, so let $\gamma=\gamma_{\min} =
\frac{\alpha^{2}(1+\delta-2\beta)}{(1-\beta)^{2}}$, and the goal is to show:\[
\frac{\sqrt{\alpha\beta}}{\sqrt{\alpha\beta+4(\alpha-\gamma_{\min})(\beta-\delta)}} >
\frac{9(\beta-\delta)}{1-2\alpha+\gamma_{\min}} - \frac{3\beta}{2\delta}+1/2,
\]
or equivalently,
\[
\frac{1}{\sqrt{1+4(1-\gamma_{\min}/\alpha)(1-\delta/\beta)}} > \frac{9(\beta-\delta)}{1-2\alpha+\gamma_{\min}} -
\frac{3\beta}{2\delta}+1/2.
\]
Replacing the term $\gamma_{\min}$ in the left-hand side, we get,
\[
\frac{1}{\sqrt{1+4(1-\frac{\alpha(1+\delta-2\beta)}{(1-\beta)^2})(1-\delta/\beta)}} >
\frac{9(\beta-\delta)}{1-2\alpha+\gamma_{\min}} - \frac{3\beta}{2\delta}+1/2.
\]
Since $\alpha > 1/2$ and $\delta > \beta^2$, replace them on the LHS, then our goal can be reduced to
\[
\frac{1}{\sqrt{3-2\delta/\beta}} > \frac{9(\beta-\delta)}{1-2\alpha+\gamma_{\min}} - \frac{3\beta}{2\delta}+1/2.
\]

Let $X=1-2\alpha+\gamma_{\min}=\frac{(1-2\alpha)(1-\beta)^2+\alpha^2(1-2\beta+\delta)}{(1-\beta)^2}$ and
$A=\frac{3\beta-2\delta}{\beta}$. So we are going to show
\[
2\delta X > \sqrt{A}(18\delta(\beta-\delta)-3\beta X + \delta X),
\]
i.e.,
\[
4\delta^2 X^2 > A(18\delta(\beta-\delta)-3\beta X + \delta X)^2,
\]
when $18\delta(\beta-\delta)-3\beta X + \delta X > 0$. The rest of the proof can be checked by using Mathematica since they are all
polynomial constraints, see Appendix~\ref{sec:ineqdelta3}.

\newpage

\appendix

\section{Mathematica Code}\label{sec:code}
\begin{verbatim}
Clear["Global`*"]

SetSystemOptions["InequalitySolvingOptions" -> "CAD" -> True];
SetSystemOptions["InequalitySolvingOptions" -> "QuadraticQE" -> False];
SetSystemOptions["InequalitySolvingOptions" -> "LinearQE" -> False];

RR = alpha > 0 && beta > 0 && alpha + beta < 1;
RR3 = alpha > 0 && beta > 0 && alpha + beta + 3*alpha*beta < 1;
Inter22 = gamma > 0 && gamma < alpha && delta > 0 && delta < beta &&
          1 - 2*beta + delta - gamma > 0 && 1 - 2*alpha + gamma - delta > 0;
l25 = gamma > 0 && gamma < alpha^2 && delta > 0 &&
      delta < beta^2 && (1 - alpha)^2*delta + beta^2*(2*alpha - 1 - gamma) < 0 &&
      (1 - beta)^2*gamma + alpha^2*(2*beta - 1 - delta) < 0;
u26 = gamma > alpha^2 && gamma < alpha && delta > beta^2 &&
      delta < beta && (1 - alpha)^2*delta + beta^2*(2*alpha - 1 - gamma) > 0 &&
      (1 - beta)^2*gamma + alpha^2*(2*beta - 1 - delta) > 0;


DD = (1 - alpha - beta)^2 + 4*(alpha - gamma)*(beta - delta);

epsilon = 1/2*(1 + alpha - beta - 2*gamma - SQ);
eta = 1/2*(1 + beta - alpha - 2*delta - SQ);

W11 = (alpha - gamma - epsilon)* epsilon*(1 - 2*alpha + gamma)/(eta*(alpha - gamma)^2);
W12 = (alpha - gamma)^2/((1 - 2*alpha + gamma)*gamma);
W21 = (beta - delta - eta)*eta*(1 - 2*beta + delta)/(epsilon*(beta - delta)^2);
W22 = (beta - delta)^2/((1 - 2*beta + delta)*delta);

R1 = (1 - alpha - beta)/(1 - alpha - beta - epsilon - eta);
R2 = SQ/(1 - 2*alpha + gamma);
R3 = 2*(alpha - gamma - epsilon)/(alpha - gamma);
R4 = SQ/gamma;
R5 = 2*(1 - beta - gamma - epsilon)/(alpha - gamma);
R6 = SQ/(1 - 2*beta + delta);
R7 = 2*(alpha - gamma - epsilon)/(beta - delta);
R8 = SQ/delta;
R9 = 2*(1 - beta - gamma - epsilon)/(beta - delta);
R10 = (1 - beta - gamma - epsilon)/(alpha - gamma - epsilon);
R11 = (1 - alpha)*(1 - beta)/(alpha*beta);
\end{verbatim}

\newpage
\subsection{Mathematica Code for Proving Lemma~\ref{lem:zeros}}\label{sec:zeros}
\begin{verbatim}
W3 = Expand[Expand[Numerator[W11]/2 - Denominator[W11]/2] /. {SQ^2 -> DD}];
Exponent[W3, SQ]
W31 = Coefficient[W3, SQ, 0]
W32 = Coefficient[W3, SQ, 1]
Resolve[Exists[{alpha, beta, gamma}, RR && Inter22 && W32 <= 0], Reals]

W4 = Expand[Expand[Numerator[W21]/2 - Denominator[W21]/2] /. {SQ^2 -> DD}];
Exponent[W4, SQ]
W41 = Coefficient[W4, SQ, 0]
W42 = Coefficient[W4, SQ, 1]
Resolve[Exists[{alpha, beta, delta}, RR && Inter22 && W42 <= 0], Reals]

W5 = FullSimplify[(W31/W32)^2 - DD]
W6 = FullSimplify[(W41/W42)^2 - DD]
\end{verbatim}

\subsection{Mathematica Code for Proving Claim~\ref{cccf}}\label{app:secondd}
\begin{verbatim}
FullSimplify[((-R1 + R2 + R3) - ((beta - delta)*(-1/(alpha - gamma - epsilon)
          - 1/eta + 1/epsilon) - SQ/epsilon + SQ/(1 - 2*alpha + gamma) +
          2*SQ/(alpha - gamma))) /. {SQ -> Sqrt[DD]}]
P1 = Expand[FullSimplify[ExpandAll[((beta - delta)*(-1/(alpha - gamma - epsilon)
          -1/eta + 1/epsilon)-SQ/epsilon + SQ/(1 - 2*alpha + gamma) +
          2*SQ/(alpha - gamma))*epsilon*eta*(1 - 2*alpha + gamma)*(alpha - gamma)*
          (alpha - gamma - epsilon)]]/.{SQ^2 -> DD, SQ^3 -> SQ*DD, SQ^4 -> DD^2}];
Exponent[P1, SQ]
P11 = Coefficient[P1, SQ, 0]
Resolve[Exists[{alpha, beta, gamma, delta}, RR && Inter22 && P11 <= 0], Reals]
P12 = Coefficient[P1, SQ, 1]
Resolve[Exists[{alpha, beta, gamma, delta}, RR && Inter22 && P12 >= 0], Reals]
\end{verbatim}

\subsection{Mathematica Code for Proving Claim~\ref{ccch}}\label{sec:ccch}
\begin{verbatim}
P2 = FullSimplify[(P11/P12)^2 - DD];
P21 = -FullSimplify[Numerator[P2]/((beta - delta)^2*(alpha - gamma)^3)]
Resolve[Exists[{alpha, beta, gamma, delta}, RR && Inter22 && P21 <= 0], Reals]
\end{verbatim}

\subsection{Mathematica Code for Proving Claim~\ref{ccccr}}\label{sec:detcoeff}
\begin{verbatim}
P5=Expand[FullSimplify[
   Expand[((-R1 + R2 + R3)*(-R1 - R8 - R9) - R1^2)*delta*(alpha - gamma)*
          epsilon* eta*(1 - 2*alpha + gamma)*(alpha - gamma - epsilon)/ SQ]/.
          {SQ -> Sqrt[DD]}]]/.{Sqrt[alpha^2 + (-1 + beta)^2 +
          alpha (-2 + 6 beta - 4 delta) + 4 (-beta + delta) gamma]->SQ}
Exponent[P5, SQ]
P51 = Coefficient[P5, SQ, 0]
Resolve[Exists[{alpha, beta, gamma, delta}, RR && Inter22 && P51 >= 0]]
P52 = Coefficient[P5, SQ, 1]
Resolve[Exists[{alpha, beta, gamma, delta}, RR && Inter22 && P52 <= 0]]
FullSimplify[DD - (P51/P52)^2]
\end{verbatim}

\subsection{Mathematica Code for Proving Claim~\ref{clm:s1}}\label{sec:s1}
\begin{verbatim}
Resolve[Exists[{alpha, beta, gamma, delta},
        FullSimplify[
        Expand[(-R2*R11 - R6*R11 + R4 + R8)/SQ*alpha*beta*gamma*
                 delta*(1 - 2*alpha + gamma)*(1 - 2*beta + delta)]] <= 0
                 && RR3 && l25], Reals]
\end{verbatim}

\subsection{Mathematica Code for Proving Claim~\ref{ckkk}}\label{sec:ckkk}
\begin{verbatim}
P7 = FullSimplify[((1 - alpha)*(1 - beta)/((1 - 2*alpha + gamma)*alpha*beta) -
                  (alpha - gamma - epsilon)/((alpha - gamma)*eta)*((1 - alpha)*
                  (1 - beta)/(alpha*beta) - (1 - beta -gamma - epsilon)
                  /(alpha - gamma - epsilon)))*alpha*beta*eta*(alpha - gamma)*
                  (1 - 2*alpha + gamma)]
Exponent[P7, SQ]
P71 = Coefficient[P7, SQ, 0]
Resolve[Exists[{alpha, beta, gamma, delta}, RR3 && Inter22 && P71 <= 0], Reals]
P72 = Coefficient[P7, SQ, 1]
Resolve[Exists[{alpha, beta, gamma, delta}, RR3 && Inter22 && P72 >= 0], Reals]
\end{verbatim}

\subsection{Mathematica Code for Proving Lemma~\ref{lem:case13}}\label{app:case13}
\begin{verbatim}
TP1 = Expand[P71^2 - P72^2*DD]
Resolve[Exists[{alpha, beta, gamma, delta},RR3 && l25 && TP1 <= 0], Reals]
\end{verbatim}

\subsection{Mathematica Code for Proving Claim~\ref{clm:qwer}}\label{sec:qwer}
\begin{verbatim}
TP2 = Expand[DD - (-(1 - alpha - beta)^2 - alpha - beta + gamma + delta)^2]
Resolve[Exists[{alpha, beta, gamma, delta}, RR && alpha <= 1/2 && beta <= 1/2
       && u26 && TP2 <= 0], Reals]
\end{verbatim}

\subsection{Mathematica Code for Proving Lemma~\ref{lem:ineqdelta3}}\label{sec:ineqdelta3}
\begin{verbatim}
X = ((1 - 2*alpha)*(1 - beta)^2 + alpha^2*(1 - 2*beta + delta))/(1 - beta)^2;
A = (2*beta - 2*delta)/beta;
Resolve[Exists[{alpha, beta, gamma, delta},
       (1 - alpha - beta)^2 == alpha*beta && alpha > 1/2 &&
        alpha < 1 && beta > 0 && beta < 1 && Inter22 && u26 &&
       4*delta^2*X^2 <= A*(18*delta*(beta - delta) - 3*beta*X + delta*X)^2 &&
       18*delta*(beta - delta) - 3*beta*X + delta*X > 0], Reals]
\end{verbatim}
\end{document}